\documentclass[12pt]{amsart}
\usepackage{amsmath,amssymb,amsthm,
color, euscript, enumerate}
\usepackage{dsfont}
\usepackage{longtable}

\newcommand{\RR}{{\mathbb R}}

\newcommand{\CC}{{\mathbb C}}
\newcommand{\beq}{\begin{equation}}
\newcommand{\eeq}{\end{equation}}
\newcommand{\ba}{\begin{array}}
\newcommand{\ea}{\end{array}}
\newcommand{\bea}{\begin{eqnarray}}
\newcommand{\eea}{\end{eqnarray}}

\usepackage{graphicx}
\usepackage{color}

\usepackage{transparent} 

\newtheorem{theorem}{Theorem} 

\newtheorem{lemma}{Lemma}[section]
\newtheorem{proposition}{Proposition}[section]
\newtheorem{definition}{Definition} [section]

\DeclareMathOperator{\sgn}{sgn}
\DeclareMathOperator{\res}{res}

\numberwithin{equation}{section}
\allowdisplaybreaks

\begin{document}

\begin{center}
{\large   \bf The direct scattering problem for perturbed Kadomtsev-Petviashvili multi line solitons} 
 
\vskip 15pt

{\large Derchyi Wu}

\vskip 5pt

{ Institute of Mathematics, Academia Sinica, 
Taipei, Taiwan}

e-mail: {\tt mawudc@gate.sinica.edu.tw}

\vskip 10pt

{\today}

\end{center}

\vskip 10pt
\begin{center}
{\bf Abstract}
\end{center}
\begin{enumerate}
\item[]{\small Regular Kadomtsev-Petviashvili II (KPII) line solitons have been investigated and classified successfully by the Grassmannians. The inverse scattering method provides a promising and powerful approach to study the stability properties of $\textrm{Gr}(N, M)_{> 0}$ KP solitons. In this paper, we complete rigorous analysis for the direct scattering problem of perturbed   $\textrm{Gr}(N, M)_{> 0}$ KP solitons}. 
\end{enumerate}

\section{Introduction}\label{S:motivation}
The   KPII   equation  
\beq\label{E:KPII-intro}
\begin{array}{c}
(-4u_{x_3}+u_{x_1x_1x_1}+6uu_{x_1})_{x_1}+3u_{{x_2}{x_2}}=0,   
\end{array}
 \eeq introduced by Kadomtsev and Petviashvili \cite{KP70}, is a two-spatial dimensional integrable generalization  of the Korteweg-de Vries (KdV) equation.  It is an asymptotic model, where $u=u(x)=u({x_1},{x_2},{x_3})$ represents the wave amplitude at the point $(x_1,x_2)$ for a fixed time $x_3$, for dispersive systems in the weakly nonlinear, long wave regime, when the wavelengths in the transverse direction are much larger than in the direction of propagation (see \cite{AS79} for a formal derivation and \cite{L03} for a rigorous one in the water wave context). 
 
 For $u(x)$ decaying at infinity,  the Cauchy problem or the well-posedness problem has been  studied intensively by many scientists via   the inverse scattering method  or  PDE techniques since late 1980's (see \cite{KS15} or \cite{Wu19} for references).
On the other hand, interesting features of the water wave can be reproduced by  the {\bf $\mathbf {\textrm{Gr}(N, M)_{\ge 0}}$ KP  solitons} which are regular KP solutions in the entire $x_1x_2$ - plane and with peaks localized and nondecaying along certain line segments and rays (see \cite{K17}, \cite{K18} for a history). They are  constructed as follows  \cite{KW13}, given  $\kappa_1<\cdots<\kappa_M$  and 
 $A=(a_{ij})\in\mathrm{Gr}(N, M)_{\ge 0}$ (full rank $N\times M$ matrices with non negative minors), 
\beq\label{E:line-tau}
\ba{c}
u_0(x)= 2\partial^2_{x_1}\ln\tau(x),
\ea
\eeq where the $\tau$-function is  the Wronskian determinant
\beq\label{E:line-grassmannian}
\begin{split}
\tau(x)=&\left|
\left(
\begin{array}{cccc}
a_{11} &a_{12} & \cdots & a_{1M}\\
\vdots & \vdots &\ddots &\vdots\\
a_{N1} &a_{N2} & \cdots & a_{NM}
\end{array}
\right)
\left(
\begin{array}{ccc}
E_{1} & \cdots & \kappa_1^{N-1}E_1\\
E_{2} & \cdots & \kappa_2^{N-1}E_2\\
\vdots & \ddots &\vdots\\
E_{M} & \cdots & \kappa_M^{N-1}E_M\\
\end{array}
\right)
\right|\\
=&\sum_{1\le j_1< \cdots< j_N\le M}\Delta_{j_1,\cdots,j_N}(A)E_{j_1,\cdots,j_N}(x),
\end{split}
\eeq where $E_j(x)=\exp\theta_j=\exp( \kappa_j x_1+\kappa_j^2 x_2+\kappa_j^3 x_3)$, the coefficients $\Delta_{j_1,\cdots,j_N}(A)$, called the Pl$\ddot{\textrm u}$cker coordinates, is the $N\times N$ minor of the matrix $A$ whose columns are labelled by the   index set $J=\{j_1< \cdots< j_N\}$, and the exponential term  $E_{j_1,\cdots,j_N}(x)$ is given by $Wr(E_{j_1},\cdots,E_{j_N})$, i.e.,
\beq\label{E:e-j1m}
E_J(x)=E_{j_1,\cdots,j_N}(x)=\Pi_{l<m}(\kappa_{j_m}- \kappa_{j_l})\exp\left( \sum_{n=1}^N\theta_{j_n}(x) \right).
\eeq   Note that the simplest   $  {\textrm{Gr}(1, 2)_{> 0}}$ KP  solitons, corresponding to $\kappa_1$, $\kappa_2$ , and $A=(1,a)$, $a>0$, are 
\beq\label{E:line-tau-oblique}
 u_0(x) 
=  \frac{(\kappa_1-\kappa_2)^2}2\textrm{sech}^2\frac{\theta_1(x)-\theta_2(x)-\ln a}2.
\eeq  
They are the $1$-line solitons discussed in literature. There has been important  
  progress in  combinatoric properties,  classification theory, and wave resonant theory of
  KPII line solitons \cite{K17}, \cite{K18} since 2000's. Different perspectives about $\mathbf {\textrm{Gr}(N, M)_{\ge 0}}$ KP  solitons, the connection to real finite gap KP solutions,   are  investigated in \cite{AG18a}, \cite{AG18b}, \cite{AG19}.

Our interest is  the stability problem of  $\mathbf {\textrm{Gr}(N, M)_{\ge 0}}$ KP  solitons.  The $H^s$-global well   posedness  of the KPII equation 
with initial data  $u_c(x_1,x_2)$ where $u_c(x_1-cx_3,x_2)$ is a KP solution has been   solved by Molinet-Saut-Tzvetkov \cite{MST11}. Their result shows that the deviation of the KPII solution from the initial data could evolve  exponentially.  Taking $
\kappa_1=-\kappa_2,\ A=(1, 1) 
$, Mizumachi establishes $L^2$-orbital stability  and $L^2$-instability theories by showing that the amplitude of the line soliton converges to that of the line soliton at
initial time whereas jumps of the local phase shift of the crest propagate in a finite speed toward $x_2 =\pm\infty$ \cite{M15}, \cite{M19}. 

An alternative approach to study the stability problem of $\textrm{Gr}(N, M)_{\ge 0}$ KP  solitons is the {\bf inverse scattering theory  (IST)} based on  the {\bf Lax pair}
\begin{equation}\label{E:KPII-lax-1}
\begin{split}
\left\{ 
{\begin{array}{l}
 (-\partial_{x_2}+\partial_{x_1}^2+u )\Phi(x,\lambda)=0,\\
 (-\partial_{x_3}+ \partial_{x_1}^3+\frac 32u\partial_{x_1}+\frac 34u_{x_1}+\frac 34\partial_{x_1}^{-1}u_{x_2}-\lambda^3   )\Phi (x,\lambda)=0 
\end{array}}
\right.
\end{split}
\end{equation}  
of the KPII equation, where $u(x)=u_0(x)+v_0(x)$ is a perturbation of the $\textrm{Gr}(N, M)_{\ge 0}$ KP  soliton $u_0(x)$ and $\Phi(x,\lambda)$ is called the Jost solution.
 Pioneering research on the IST with data nondecaying along a single line for KPII were derived by   \cite{BP301},  \cite{VA04}. For data being a perturbed $\mathbf {\textrm{Gr}(N, M)_{\ge 0}}$ KP  soliton,    Boiti-Pempenelli-Pogrebkov-Prinari introduce various methods, the Darboux transform \cite{P00}, \cite{BP301}, \cite{BP302},  the twisting transformations \cite{BP309}, \cite{BP310}, and the $\tau$-function formulation \cite{BP310}, \cite{BP214} to set  foundations for the IST for KPII. Most significant achievements include deriving   explicit formula of the Green function \cite{BP211}, \cite{BP212},    boundedness of $G_d$ (the discrete summand in the Green function) \cite{BP211}, \cite{BP212}, and deriving the $\mathcal D^\flat $-symmetry     of the eigenfunctions \cite[(4.38)]{BP214}.

In the previous works \cite{Wu18},  \cite{Wu19},  via a KdV theory approach  and a Sato theory approach,   the direct scattering problem  for a perturbed $\textrm{Gr}(1,2)_{> 0}$ KP soliton, $u(x)=u_0(x)+v_0(x)$,  is rigorously   completed  by establishing a $\lambda$-uniform estimate of the Green function $G$ of the heat operator $-\partial_{x_2}+\partial^2_{x_1}+2\lambda\partial_{x_1}+u _0(x)$, 
\beq\label{E:intro-green-eq}
| G(x,x',\lambda)|\le C (1+\frac1{\sqrt{|x_2-x_2'|}} ),\quad \textit{$C$  a constant}
\eeq the forward scattering transformation \cite[Definition 3.2]{Wu19}
\beq\label{E:scaatering}
\mathcal S(u(x))=\{0,\kappa_1,\kappa_2,s_d,s_c(\lambda)\},
\eeq  and a Cauchy integral equation (or a $\overline\partial$ equation) with controllable singularities, 
\begin{gather}
  {  m}(x, \lambda) =1+\frac{  m_{ \res }(x  )}{\lambda  }  +\mathcal CT
   m , \label{E:intro-cauchy-eq}\\
   \quad m\in W , \label{E:intro-cauchy-eq-new}
\end{gather}
(see \cite{Wu18},  \cite{Wu19} for definitions of $m_{\res}$, $\mathcal S$, $T$, $W$). We remark that, with the help of \eqref{E:intro-green-eq},   a Picard iteration can be utilized to solve the  eigenfunction  of the heat equation $(-\partial_{x_2}+\partial^2_{x_1}+2\lambda\partial_{x_1}+u  (x))m =0$. Besides,  the discrete scattering data $0$, $\kappa_j$ are determined by the Sato eigenfunction (see \eqref{E:sato}), the $\textrm{Gr}(1,2)_{> 0}$ KP soliton; and the continuous scattering data $s_c(\lambda)$ is a nonlinear Fourier transform of the initial perturbation $v_0(x)$. Finally, \eqref{E:intro-cauchy-eq} and  \eqref{E:intro-cauchy-eq-new}, a dynamic reformulation of the forward scattering transformation \eqref{E:scaatering},  serve as an analogue of  the Gelfand-Levitan-Marchenko equation for the KdV equation in solving the inverse problem.

The goal of this paper is
 to  demonstrate there is a strong analogy between rigorous analysis of the direct problem  for perturbed $\textrm{Gr}(1,2)_{> 0}$   and   perturbed  $\textrm{Gr}(N, M)_{> 0}$ KP solitons. The key observation is that, based on a deep Sato theory, the kernel, defining the Green functions for a $\textrm{Gr}(N,M)_{> 0}$ KP soliton, is found to share  similar singular structures as that for a $\textrm{Gr}(1,2)_{> 0}$ KP soliton.
  Hence one can adapt techniques  in \cite{Wu19} to derive corresponding  \eqref{E:intro-green-eq} and \eqref{E:intro-cauchy-eq} for perturbed   $\textrm{Gr}(N, M)_{> 0}$ KP solitons. In defining corresponding forward scattering transformation $\mathcal S$ and eigenfunction space $W$, difficulties occur in search for a notion generalizing the norming constant  $s_d$ in \eqref{E:scaatering}. This obstruction has been removed by the $\mathcal D^\flat$-symmetry introduced in \cite[(4.38)]{BP214}. Heavily building on the $\mathcal D^\flat$-symmetry, and taking it in directions for being linearisable and compatible with a Cauchy integrable equation possessing controllable singularities,  we generalize norming constants and define corresponding $\mathcal S$, $W$.

The contents of the paper are as follows. In Section \ref{S:eigen}, we justify the unique solvability of an associated heat equation with proper boundary value (see \eqref{E:Lax-bdry}). Precisely, a simplified argument of Boiti et al \cite{BP214} to derive an explicit  form of the Green function will be provided first. Then we will prove a $\lambda$-uniform estimate for the Green function, Proposition \ref{P:eigen-green},  which yields the existence of the  eigenfunction $m(x,\lambda)$.
 
  In Section \ref{S:sd}, with the help of the $\mathcal D^\flat$-symmetry, we shall characterize analytic and algebraic constraints of the  eigenfunction $m(x,\lambda)$. To formulate a Cauchy integral equation with controllable singularities, a special renormalization,   $\widetilde m(x,\lambda)$,  is introduced to resolve the multiplicity of poles property of $m(x,\lambda)$ caused by the Sato eigenfunction in \eqref{E:Lax-bdry}. Then we extract the scattering data $\widetilde s_c$, $\widetilde {\mathcal D}$, define the forward scattering transform $\mathcal S$, and prove that the scattering data can linearize the  Kadomtsev-Petviashvili equation \eqref{E:KPII-intro}.

In Section  \ref{S:CIO}, we define the eigenfunction space $W$ based on properties of $\widetilde m$. Then, for $x$ fixed, after providing important $L^\infty$ estimates on $\mathcal CT\widetilde m$ (Theorem \ref{T:bdd}), we will derive a Cauchy integral equation with controllable singularities for $\widetilde m(x,\lambda)$ (Theorem \ref{T:cauchy-integral-eq}). Together with the algebraic constraints in $W$, a closeness property is addressed at the end of the paper. 

 
Throughout this paper, we set  $x_3=0$ and $x=(x_1,x_2)$ unless special mention and denote $C$ various uniform constants which are independent of $x$, $\lambda$. For convenience, we provide a table of notations here.
\vskip.5in
{\tiny\begin{center}
{\bf \large   Notation and Terminology}
\vskip.1in
\begin{tabular}{|c|c|c|}
   \hline
\bf{Notation}& \bf{Explanation} & \bf{Examples }  \cr
& & \bf{in text}\cr
   \hline\hline
   
 $\tau(x),\,\kappa_j,\,A,\,\theta_j(x), $&definition of the tau function; & \eqref{E:line-tau}-\eqref{E:e-j1m}\cr
$J,\,\Delta_J(A),\,E_J(x)$;& & \cr
$\mathbf {\textrm{Gr}(N, M)_{\ge 0}}$ KPII solitons  &regular multi-line soliton &\cr
\hline

$ u(x)= u_0(x)+ v_0(x)$  & $\qquad $ initial data with $u_0$ a $\mathbf {\textrm{Gr}(N, M)_{\ge 0}}$ KPII soliton $\quad\ $ &Theorem \ref{E:exitence-spectral}-\ref{T:cauchy-integral-eq}  \cr 
& and $v_0$ a perturbation&\cr \hline
$ \mathcal L; L$  & the heat operator associated to $u_0$    &\eqref{E;spectral} \cr\hline
$\ba{c }\mathcal G(x,x',\lambda),\,  G (x,x',\lambda);  \ea$ & Green  functions; & Definition \ref{D:green};    \cr 
$\ba{c} \mathcal G_c(x,x',\lambda),\,\mathcal G_d(x,x',\lambda),\\
    G_c(x,x',\lambda),\,G_d(x,x',\lambda);\ea$ &${\ba{c} \textrm{continuous Green functions,}\\ \textrm{discrete Green  functions}\ea}$
     & \eqref{E:green-heat};   \cr 
 \hline

$\theta(s),\,\delta(s)$  & $\qquad$ Heaviside function, Dirac function $\qquad$  & Lemma \ref{L:green-heat}, \cr
   &&Definition \ref{D:green} \cr \hline

\end{tabular}
\end{center} }

{\tiny\begin{center}
\begin{tabular}{|c|c|c|}
\hline   
$\varphi(x,\lambda),\,\chi(x,\lambda);$ & Sato eigenfunction, Sato normalized eigen- &\eqref{E:sato}; \cr
$\varphi_j(x),\,\chi_j(x)$  &function;  values of Sato  eigenfunction  at $\kappa_j$, & \eqref{E:residue-eigenfunction}   \cr
   & values of Sato normalized eigenfunction  at $\kappa_j$&   \cr\hline 
  
$\psi(x,\lambda),\,\xi(x,\lambda);$ & Sato adjoint eigenfunction, Sato normalized adjoint &\eqref{E:pole}; \cr
$\psi_j(x),\,\xi_j(x)$  & eigenfunction;  the residue of Sato adjoint eigenfunction & \eqref{E:residue-eigenfunction}  \cr
   & at $\kappa_j$, the residue of normalized  adjoint  eigenfunction at $\kappa_j$ &   \cr\hline  
   $\Phi(x,\lambda),\,m(x,\lambda)$; &  Jost solution, eigenfunction; renormalized   & \eqref{E:Lax-bdry}, \eqref{E:renormal}; \cr
   $\widetilde\Phi(x,\lambda),\,\widetilde m(x,\lambda)$; & eigenfunction, $\widetilde\Phi(x,\lambda)=e^{\lambda x_1+\lambda^2x_2}\widetilde m(x,\lambda)$;  & Definition \ref{D:reg-m};    \cr
$\kappa_j^+$; $\widetilde\Phi_j$; $\widetilde\Phi(x,\kappa)$ &$\ $ $\kappa_j^+=\kappa_j+0^+$; $\widetilde\Phi_j=\widetilde\Phi(x,\kappa_j^+)$; $\widetilde\Phi(x,\kappa)=(\widetilde\Phi_1,\cdots,\widetilde\Phi_M)$ $\ $ & \eqref{E:m-sym-cond}; \eqref{E:reg-m-sym}\cr
    \hline 
$\mathcal D,\,\mathcal D';$ $\mathcal D^\flat;\,\widetilde{\mathcal D} $&  symmetries associated to Sato eigenfunction, &Definition \ref{D:dd'}\cr
        & and to  Sato adjoint eigenfunction; symmetries & \eqref{E:m-sym};\cr
      &associated to $\Phi$; symmetries associated to $\widetilde\Phi$&\eqref{E:sharp}, \eqref{E:reg-m-sym} \cr\hline 
      
$ G\ast f $ & convolution operator & Definition \ref{D:con} \cr \hline
 $m_{\res,n}(x),\,m_{0,r}(x,\lambda)$;& coefficients of the Laurent series of $m$   at $\lambda=0$;  & \eqref{E:pole-m} \cr
 
 $\widetilde m_{z_n,\res}(x),\,\widetilde m_{z_n,r}(x,\lambda)$&    residue and remainder of $\widetilde m$    at $\lambda=z_n$  & \eqref{E:define-remainder-m-sf-pm2} \cr
 \hline   
   
$\quad$ $D_z,\,D_z^\times,\, D_{z,r},\,D_{z,r}^\times$;$\quad$ & $\quad$$\quad$ disks and punctured disks at $z$;  characteristic $\quad$$\quad$& Definition \ref{D:terminology} \cr
$E_z(\lambda),\,E_A( \lambda)$; &    functions of $D_z$, and of a     non single point&   \cr
$ a$; $\tilde a$; & set $A$;  $ a=\frac 12\min_{j\ne j'}\{|\kappa_j|, \,|\kappa_j-\kappa_{j'}|\}$; $\tilde a= a/N$;  &   \cr 
$\widetilde D_z,\,\widetilde D_z^\times$; & disks and punctured disks at $z$ with radius $\tilde\kappa$   & Definition \ref{D:reg-m} \cr\hline      
$\mathfrak G_j(x,x'); $&   leading  continuous of $G$ at $\lambda=\kappa_j$; & \eqref{E:g-asym-pm-i-prep-limit-sym} \cr 
 $  \Theta_j(x),\,\gamma_j$& factors of leading discontinuities of $m$ at $\lambda=\kappa_j$ & \eqref{E:discon} \cr \hline
$  s_c(\lambda), \tilde s_c(\lambda);\ z_n;   $& continuous scattering data; poles of $\widetilde m$; & \eqref{E:conti-sd}, \eqref{E:tilde-dbar-m-sf-new-0};   \cr
 $T$, $\mathcal S$  & continuous scattering operator,  & Definition \ref{D:reg-m};    \cr
 &  forward scattering transform  &  Definition \ref{D:spectral}\cr   \hline 
   
 $\mathcal C=\mathcal C_\lambda$ & Cauchy integral operator &  Definition \ref{D:cauchy}\cr 
 \hline   
 
 $W$ & the eigenfunction space  & Definition \ref{D:quadrature-hat}  \cr  
\hline

\end{tabular}
\end{center} }

{\bf Acknowledgments}. We feel deeply indebted to A. Pogrebkov and Y. Kodama   for introducing  the  Sato theory of the  KP hierarchy.  Special thanks need to be expressed to A. Pogrebkov for an illuminating explanation on the boundedness of $G_d$.  We would like to pay  respects to the fundamental contribution on  the inverse scattering theory done by Boiti, Pempinelli, Pogrebkov, Prinari. This research project was  partially supported by NSC 107-2115-M-001 -002 -.

\section{Direct Problem : the eigenfunction} \label{S:eigen}

The direct scattering problem starts with finding the Jost solution solving an associated heat equation with proper boundary value (see \eqref{E:Lax-bdry}). In this section, we shall simplify the argument of Boiti et al \cite{BP214} to derive an explicit  form of the Green function first. Then we provide a $\lambda$-uniform estimate for the Green function which yields the existence of the Jost solution.
 
Let $u(x)=u_0(x)+v_0(x)$, $u_0(x)=2\partial^2_{x_1}\ln\tau(x)$ be a $\textrm{Gr}(N, M)_{\ge  0}$ KP soliton defined by the data $\kappa_1<\cdots<\kappa_M$,  
 $A=(a_{ij})\in\mathrm{Gr}(N, M)_{\ge 0}$ through \eqref{E:line-tau}, \eqref{E:line-grassmannian}.  Consider the boundary value problem
\beq\label{E:Lax-bdry}
\left\{
{\ba{ll}
 (-\partial_{x_2}+\partial_{x_1}^2 
+u(x))\Phi(x ,\lambda)=0,& \\
\lim_{|x|\to\infty}e^{-(\lambda x_1+\lambda^2x_2)}(\Phi(x,\lambda)-\varphi(x,\lambda))=0,
\ea}
\right. 
\eeq  for any fixed  $ \lambda\ne 0$, $\lambda_R\notin  \{   \kappa_1,\cdots, \kappa_M  \}$, where 
\beq\label{E:sato}
\begin{split}
&\varphi(x,\lambda)\\
=&e^{\lambda x_1+\lambda^2 x_2}\frac{\sum_{1\le j_1< \cdots< j_N\le M}\Delta_{j_1,\cdots,j_N}(A) (1-\frac{\kappa_{j_1}}\lambda)\cdots(1-\frac{\kappa_{j_N}}\lambda)E_{j_1,\cdots,j_N}(x)}{\tau(x)} \\
\equiv&e^{\lambda x_1+\lambda^2 x_2}\chi(x,\lambda)
\end{split}
  \eeq  is the {\bf Sato eigenfunction} and $\chi(x,\lambda)$ is the {\bf Sato normalized eigenfunction} \cite[(2.12)]{BP214}, \cite[Theorem 6.3.8., (6.3.13) ]{D91}, \cite[Proposition 2.2, (2.21)]{K17} satisfying 
\beq\label{E;spectral}
\begin{split}
&\mathcal L\varphi(x,\lambda)\equiv \left(-\partial_{x_2}+\partial^2_{x_1}+u_0(x)\right)\varphi(x,\lambda)=0,\\
&L\chi(x,\lambda)\equiv \left(-\partial_{x_2}+\partial^2_{x_1}+2\lambda\partial_{x_1}+u_0(x)\right)\chi(x,\lambda)=0.
\end{split}
\eeq   
Renormalizing  $\Phi(x,\lambda)=e^{\lambda x_1+\lambda^2x_2} m(x,\lambda)$,  the boundary value problem \eqref{E:Lax-bdry} turns  into
\beq\label{E:renormal}
\left\{
{\ba{l}
  Lm(x,\lambda)=-v_0(x)m(x,\lambda), \\
 \lim_{|x|\to\infty}(m(x,\lambda) -\chi(x,\lambda))=0 ,
\ea}
\right. 
\eeq for any fixed  $ \lambda\ne 0$, $\lambda_R\notin  \{   \kappa_1,\cdots, \kappa_M  \}$.

\begin{definition}\label{D:green}
Define the Green functions, associated to a $\textrm{Gr}(N, M)_{\ge  0}$ KP soliton $u_0(x)$, by $\mathcal G(x,x', \lambda)$ and $G(x,x', \lambda)$ satisfying
\beq\label{E:sym-0}
\begin{gathered}
\mathcal L\mathcal G(x,x', \lambda) =\delta(x-x') ,\ \
L  G(x,x', \lambda)=\delta(x-x'),\\ 
\mathcal G(x,x', \lambda) =e^{ \lambda (x_1-x_1')+\lambda^2 (x_2-x_2') } G(x,x', \lambda).
\end{gathered}
\eeq Here $\delta(x)$ is the dirac function at $x=0$.
\end{definition}
Following the approach of Boiti et al \cite{BP214}, the key ingredient to derive the Green function is  an orthogonality relation between Sato  eigenfunctions  and {\bf Sato   adjoint eigenfunctions}, defined by
\beq\label{E:pole}
\begin{split}
&\psi(x,\lambda)\\
=&e^{-(\lambda x_1+\lambda^2 x_2)}\frac{\sum_{1\le j_1< \cdots< j_N\le M}\Delta_{j_1,\cdots,j_N}(A)\frac {E_{j_1,\cdots,j_N}(x)}{(1-\frac{\kappa_{j_1}}\lambda)\cdots(1-\frac{\kappa_{j_N}}\lambda)}}{\tau(x)} \\
\equiv & e^{-(\lambda x_1+\lambda^2 x_2)}\xi(x,\lambda). 
\end{split}
\eeq   Here $\xi(x,\lambda)$ is called the {\bf Sato normalized adjoint eigenfunction} \cite[(2.12)]{BP214}, \cite[Theorem 6.3.8., (6.3.13) ]{D91}.   Note  
\beq\label{E;spectral-new}
\begin{split}
&\mathcal L^\dagger\psi(x,\lambda)\equiv\left(\partial_{x_2}+\partial^2_{x_1}+u_0(x)\right)\psi(x,\lambda)=0,\\
&L^\dagger\xi(x,\lambda)\equiv \left(\partial_{x_2}+\partial^2_{x_1}-2\lambda\partial_{x_1}+u_0(x)\right)\xi(x,\lambda)=0.
\end{split}
\eeq Moreover, for $\forall x\in\RR^2$ fixed, $\chi(x,\cdot)$   is a rational function  normalized at $ \infty$ and with a  pole at $ 0$ of multiplicity $N$; and $\xi(x,\cdot)$ is a rational function  normalized at $ \infty$  with a  zero at $0$ of multiplicity $N$, and simple poles at $\kappa_1$, $\cdots$, $\kappa_M$. Therefore, values of $\varphi$ and residues of $\psi$ at $\kappa_j$ are significant in deriving the orthogonality relation. Let
\beq\label{E:residue-eigenfunction}
\begin{gathered}
\varphi_j(x)=\varphi(x,\kappa _j)=e^{\kappa_j x_1+\kappa_j^2x_2} \chi_j(x),
\\
\psi_j(x)=\textit{res}_{\lambda=\kappa_j}\psi(x,\lambda)=e^{-(\kappa_j x_1+\kappa_j^2x_2)}\xi_j(x),
\end{gathered}
\eeq and
\beq\label{E:residue-eigenfunction-M}
\begin{split}
&\varphi(x,\kappa)= (\varphi_1(x), \cdots, \varphi_M  (x)),\\
&\psi(x,\kappa)= (\psi_1(x), \cdots, \psi_M  (x)).
\end{split}
\eeq

\begin{definition}\label{D:dd'} For a $Gr(N, M)_{\ge 0}$ KP soliton $u_0(x)$, write  
\beq\label{E:D}
 A = \left(\ba{c} 
I_N,\ d \ea \right)\pi,
\eeq where $\pi$ is an $M\times M$ permutation matrix and $d$ is an $N\times (M-N)$ matrix. 
 Define
\beq\label{E:D'}
\begin{split}
&\mathcal D= \textrm{diag}\,(   
\kappa^N_1 ,\cdots,\kappa^N_M )\, A^T ,\\
&\mathcal D' = \left(\ba{c}-d^T,\ I_{M-N}\ea\right)\,\pi\,\textrm{diag}\,(   
\kappa^{-N}_1 ,\cdots,\kappa^{-N}_M  ).
\end{split}
\eeq
\end{definition} In \cite{BP310}, Boiti  et al prove 
the following   $\mathcal D$, $\mathcal D'$ symmetries about the Sato eigenfunction and the   Sato adjoint eigenfunction for $Gr(N, M)_{\ge 0}$ KP solitons and  an orthogonality relation between $\mathcal D$ and $\mathcal D'$. They yield an explicit form of the Green function and will induce full symmetries of $m(x,\lambda)$ (see \eqref{E:m-sym-cond}). In the following lemma we provide a direct and simplified proof via an approach suggested by Yuji Kodama. 
\begin{lemma}\label{L:BP3-symmetry} \cite{BP310} Define $\mathcal D$ and $\mathcal D'$ by Definition \ref{D:dd'}. Then
\beq\label{E:D-sym}
\mathcal D'\mathcal D=0,
\ \ 
\varphi(x,\kappa)\mathcal D=0,
\ \ 
\mathcal D'\psi(x,\kappa)^T=0.
\eeq 

\end{lemma}
\begin{proof}

  From Definition \ref{D:dd'},
\begin{align*}
 \mathcal D'\mathcal D 
=&\left(\ba{c}-d^T,\ I_{M-N}\ea\right)\,\pi\,\textrm{diag}\,(   
\kappa^{-N}_1 ,\cdots,\kappa^{-N}_M  )
\textrm{diag}(   
\kappa^N_1 ,\cdots,\kappa^N_M )\pi^T \left(\ba{c} 
I_N\\d^T\ea \right) \\
=&\left(\ba{c}-d^T,\ I_{M-N}\ea\right) 
\left(\ba{c} 
I_N\\d^T\ea \right)  =0.
\end{align*}

To prove  $\varphi\mathcal D=0$, writing
\[
A =\left(
\begin{array}{cccc}
a_{11} &a_{12} & \cdots & a_{1M}\\
\vdots & \vdots &\ddots &\vdots\\
a_{N1} &a_{N2} & \cdots & a_{NM}
\end{array}
\right), 
\]and using \eqref{E:e-j1m}, for $1\le n\le N$, we obtain  
\begin{align*}
&\textit{the $n$-th column of }\varphi(x,\kappa)\mathcal D\\
=&\sum_{m=1}^M\varphi_m\kappa_m^N a_{nm}\\
=&\frac 1{\tau(x)} \sum_{m=1}^M a_{nm} \sum_{ 
  j_1, \cdots, j_N   } \Pi_{k=1}^N(\kappa_m-\kappa_{j_k})\Pi_{\beta<\alpha}(\kappa_{j_\alpha}-\kappa_{j_\beta})\Delta_{j_1,\cdots,j_N}(A)   e^{\theta_{j_1}  +\cdots+\theta_{j_N}+\theta_m} \\
=&\frac {(-1)^{N-1}}{\tau(x)}  \sum_{k=1}^{N+1}\sum_{ 
  j'_1, \cdots, j'_{N+1}   }  \Pi_{\beta<\alpha}(\kappa_{j'_\alpha}-\kappa_{j'_\beta})\Delta_{j'_1,\cdots,\check j'_k,\cdots,j'_{N+1}}(A)  (-1)^{k} a_{nj'_k} e^{\theta_{j'_1}  +\cdots+\theta_{j'_{N+1}}}\\
  =&\frac {(-1)^{N-1}}{\tau(x)}  \sum_{ 
  j'_1, \cdots, j'_{N+1}   } E_{j_1',\cdots,j_{N+1}'}\left[\sum_{k=1}^{N+1}\Delta_{j'_1,\cdots,\check j'_k,\cdots,j'_{N+1}}(A)  (-1)^{k} a_{nj'_k}\right]=0.
\end{align*} 
Here $(i_1,\cdots,\check i_{k}  ,   \cdots,i_N)\equiv (i_1,\cdots,i_{k-1}  , i_{k+1}  ,\cdots,i_N)$ and for the last equality, we have used the following Pl$\ddot{\textrm u}$cker relation
\begin{align*}
0=&\left|\ba{lclcll}
a_{1j_1} & \cdots & a_{nj_1} &\cdots & a_{Nj_1}   & a_{nj_1}\\
a_{1j_2} & \cdots & a_{nj_2} &\cdots & a_{Nj_2}   & a_{nj_2}\\
\vdots & \ddots & \vdots &\ddots & \vdots  & \vdots\\
a_{1j_{N+1}} & \cdots & a_{nj_{N+1}} &\cdots & a_{Nj_{N+1}}   & a_{nj_{N+1}}
\ea\right|\\
=&(-1)^{N-1}\sum_{k=1}^{N+1}\Delta_{j'_1\cdots\check{j'_k}\cdots j'_{N+1}}(A)(-1)^ka_{nj'_k}.
\end{align*}Hence the $\mathcal D$-symmetry of the Sato eigenfunction  $\varphi \mathcal D=0$ is verified. 

For the $\mathcal D'$-symmetry,  write 
\[
B\equiv\left(\ba{c}-d^T, I_{M-N}\ea\right)\pi=\left(\ba{cccc} b_{11} &b_{12} &\cdots &b_{1M}\\
\vdots &\vdots &\ddots &\vdots\\
 b_{(M-N),1}  & b_{(M-N),2}  &\cdots &b_{(M-N),M}\ea\right),
 \] then  
\beq\label{E:plucker-coord-duality}
\begin{gathered}
\Delta_I(B)=\sigma(J,I)\Delta_J(A),\\
I=\{i_1<i_2<\cdots<i_{M-N}\},\, J=\{j_1<j_2<\cdots<j_N\},\\ I\cup J=\{1,\cdots,M\},\ \
\sigma(J,I)=(-1)^{\frac {N(N+1)}2+j_1+\cdots+j_N} 
\end{gathered}
\eeq \cite[\S 4.4, pp 100-101]{K18}. Thus, for $1\le n\le M-N$,
\begin{align}
&\textit{the $n$-th row of }\mathcal D'\psi^T(x,\kappa)\label{E:plucker-adj}\\
=&\sum_{m=1}^M\psi_m\kappa_m^{-N} b_{nm}\nonumber\\
=&\frac 1{\tau(x)} \sum_{m=1}^M b_{nm} \sum_{ 
 m=j_k\in\{ j_1, \cdots, j_N\}   }(-1)^{N-k}\Delta_{j_1,\cdots,j_N}(A)\nonumber\\
 \times& \Pi_{j_1, \cdots,\check m,\cdots, j_N;\,\beta<\alpha}(\kappa_{j_\alpha}-\kappa_{j_\beta})   e^{\theta_{j_1}  +\cdots+\check{\theta}_m+\cdots+\theta_{j_N} }\nonumber\\
=&\frac { (-1)^{\frac{N(N+1)}{2}}}{\tau(x)}  \sum_{ 
  j'_1, \cdots, j'_{N-1}   } (-1)^{ j'_1+\cdots+j'_{N-1}} \Pi_{\beta<\alpha}(\kappa_{j'_\alpha}-\kappa_{j'_\beta})e^{\theta_{j'_1}  +\cdots+\theta_{j'_{N-1}}}\nonumber\\
   \times&\sum_{k=1}^{M-N+1}\Delta_{i'_1,\cdots,  \check i'_k,\cdots,i'_{M-N+1}}(B)  (-1)^{k} b_{ni'_k} =0.\nonumber
\end{align} Here $I'=\{i'_1 <\cdots<i'_{M-N+1}\}$, $ J'=\{j'_1 <\cdots<j'_{N-1}\}$, $I'\cup J'=\{1,\cdots,M\}$, and we have used the Pl$\ddot{\textrm u}$cker relation
\begin{align*}
0=&\left|\ba{lclcll}
b_{1i'_1} & \cdots & b_{ni'_1} &\cdots & b_{M-N,i'_1}   & b_{ni'_1}\\
b_{1i'_2} & \cdots & b_{ni'_2} &\cdots & b_{M-N,i'_2}   & b_{ni'_2}\\
\vdots & \ddots & \vdots &\ddots & \vdots  & \vdots\\
b_{1i'_{M-N+1}} & \cdots & b_{ni'_{M-N+1}} &\cdots & b_{M-N,i'_{M-N+1}}   & b_{ni'_{M-N+1}}
\ea\right|\\
=&(-1)^{N-1}\sum_{k=1}^{M-N+1}\Delta_{i'_1,\cdots,  \check i'_k,\cdots,i'_{M-N+1}}(B)  (-1)^{k} b_{ni'_k}.
\end{align*}

\end{proof}

\begin{lemma}\label{L:reside}  \cite{BP309}, \cite{BP310}, \cite{BP211-tmp}, \cite{BP211-jmp}  The Sato eigenfunction $\varphi$ and adjoint eigenfunction $\psi$, associated to a $Gr(N, M)_{\ge 0}$ KP soliton $u_0(x)$,  satisfy the {\bf orthogonality relation}
\[
\sum_{j=1}^M\varphi_j (x)\psi_j(x')=0.
\]
\end{lemma}
\begin{proof} The lemma follows by applying Lemme \ref{L:BP3-symmetry} and   the projection
\beq\label{E:d-sym-further}
\begin{split}
&P=\mathcal D(\mathcal D^T\mathcal D)^{-1}\mathcal D^T,\\
&P'=(\mathcal D')^T(\mathcal D'{\mathcal D'}^T)^{-1}\mathcal D',
\end{split}
\eeq which are orthogonal projectors, i.e.,
\beq\label{E:ortho}
\begin{gathered}
P^2=P,\quad(P')^2=P',\quad PP'=0=P'P,\\
P+P'=I_{M\times M}.
\end{gathered}
\eeq

\end{proof}
 
\begin{lemma}\label{L:green-heat} \cite{BP214} Let $\theta$ be the Heaviside function and $\lambda=\lambda_R+i\lambda_I$, $\lambda_R,\,\lambda_I\in\RR$. Then the Green function, associated to a $Gr(N, M)_{\ge 0}$ KP soliton $u_0(x)$,  has the explicit form
\begin{align}
&\mathcal G(x,x',\lambda)=\mathcal G_c(x,x',\lambda)+\mathcal G_d(x,x',\lambda),\label{E:green-heat} \\
&\mathcal G_c= -\frac{\sgn(x_2-x_2')}{2\pi}\int_\RR\theta((s^2-\lambda_I^2)(x_2-x_2')) \varphi(x,\lambda_R+is)\psi(x',\lambda_R+is) ds,\nonumber\\
&\mathcal G_d=-\theta(x_2'-x_2)\sum_{j=1}^M\theta(\lambda_R-\kappa_j)\varphi_j(x)\psi_j(x') .
\nonumber
\end{align}

\end{lemma}

\begin{proof} After deriving the orthogonality relation, Lemma \ref{L:reside}, one can apply Fourier inversion formula, meromorphic properties of normalized eigenfunctions $\chi(x,\lambda)$ and $\xi(x,\lambda)$, and the residue theorem to derive the explicit formula of the Green function $\mathcal G(x,x',\lambda)$. For a direct proof, we refer to \cite[Lemma 2.2]{Wu19}.

\end{proof}

\begin{definition}\label{D:terminology} 
For $z\in \mathfrak Z=\{ 0,\, \kappa_1,\cdots, \kappa_M  \}$,  define 
 {$a=\frac 12\min_{j\ne j'}\{|\kappa_j|, \,|\kappa_j-\kappa_{j'}|\}$},
\[
\ba{c}
D_{z}=\{\lambda\in\CC\, :\, |\lambda-z|< { a}\},\ 
   D_z^\times=\{\lambda\in\CC\, :\, 0<|\lambda-z|<  a\};\\
D_{z,r}=\{\lambda\in\CC\, :\, |\lambda-z|<  r  a\},\ 
   D_{z,r}^\times=\{\lambda\in\CC\, :\, 0<|\lambda-z|< r  a\} ,   
 \ea
 \]  characteristic functions  $\mathcal E_z(\lambda)\equiv 1$ on $D_z$, $\mathcal E_z (\lambda)\equiv 0$ on $D_z^c$, $\mathcal E_A(\lambda)\equiv 1$ on $A$, $\mathcal E_A (\lambda)\equiv 0$ elsewhere.
  \end{definition}  
  
\begin{definition}\label{D:con} For any Schwartz function $ f(x)$, define the convolution  operator 
\beq\label{E:convolution}
G\ast f(x,\lambda)\equiv\iint  G(x,x',\lambda)f(x' )dx' ,\ \ dx'=dx'_1dx'_2.
\eeq

\end{definition}   
So  the boundary value problem \eqref{E:renormal} turns into
\beq\label{E:spectral-integral}
m(x,\lambda)=\chi(x,\lambda)-  G\ast v_0m (x,\lambda) 
\eeq       
 formally. Since the initial potential $u(x)=u_0(x)+v_0(x)$ is independent of $\lambda$. To utilize a Picard iteration to solve the above integral equation for $m(x,\lambda)$,   a $\lambda$-uniform estimate for the Green function $G$ is necessary. We provide a $\lambda$-uniform estimate in the following proposition. Note that, throughout the paper, we use $C$ to denote different uniform constants.
\begin{proposition}\label{P:eigen-green}
The Green function $ G$, associated to a {\bf totally positive $Gr(N, M)_{> 0}$ KP soliton $u_0(x)$},   satisfies, for $\forall x_2-x_2'\ne 0$,   $\lambda\ne \kappa_1,\cdots,\kappa_M$,    
\beq\label{E:eigen-green}
| G(x,x',\lambda)|\le C (1+\frac1{\sqrt{|x_2-x_2'|}} ),   
\eeq    
 and 
for any Schwartz function $ f$, 
\beq\label{E:ast}
\begin{gathered}
\lim_{|x|\to\infty}G(x,x',\lambda)\ast f(x')  \to 0 .
\end{gathered}
\eeq  

\end{proposition}

\begin{proof} 
 From \eqref{E:sato},  \eqref{E:pole}, and Lemma \ref{L:green-heat},   
\begin{align}
&  G_c (x,x', \lambda)\label{E:g-c-gr-nm}\\
= &-\frac 1{2\pi}\int_\RR ds\ \mbox{sgn}(x_2-x_2')\theta((s^2-\lambda_I^2)(x_2-x_2')) \chi(x,\lambda_R+is)\nonumber\\
  \times & \xi(x',\lambda_R+is) e^{ [\lambda-(\lambda_R+is)] (x_1'-x_1)+[\lambda^2-(\lambda_R+is)^2](x_2'-x_2) }\nonumber\\
=&-\frac{e^{i[\lambda_I(x_1'-x_1)+2\lambda_I\lambda_R(x'_2-x_2)]}}{2\pi} \int_\RR ds\ { \mbox{sgn}(x_2-x_2')\theta((s^2-\lambda_I^2)(x_2-x_2')) } \nonumber\\
 \times &  e^{ { (s^2-\lambda^2_I)(x_2'-x_2)}- { is[  (x_1'-x_1)+2 \lambda_R  (x_2'-x_2)] }}\nonumber \\
  \times & \sum_J\frac{\Pi_{m\in J}({\lambda_R+is}-{\kappa_m})\Delta_J(A)E_J(x)}{\tau(x)}\, \sum_I\frac{\Delta_I(A)E_I(x') }{\Pi_{n\in I}({\lambda_R+is}-{\kappa_n} )\tau(x')},\nonumber
\end{align} where $J=\{j_1< \cdots< j_N\}$,  $I=\{i_1< \cdots< i_N\}$, $1\le j_m, i_n\le M$, and $\Delta_J(A)$, $\Delta_I(A)$ are the Pl$\ddot{\textrm u}$cker coordinates (see \eqref{E:line-grassmannian} and \eqref{E:e-j1m}).    
Since  
\[
\Delta_J(A), \ \Delta_I(A),\ |\frac{E_J(x)}{\tau(x)}|,\  |\frac{E_I(x')}{\tau(x')}|<C,
\] replacing $2$ by general $ M\ge 2$,  \eqref{E:g-c-gr-nm} is in a complete analogy with \cite[ (2.21), (2.30), (2.31) ]{Wu18} and  one can follow   \cite[Proposition 2.1]{Wu18} to decompose \eqref{E:g-c-gr-nm} into a combination of discontinuous, oscillatory, exponentially decaying factors  and derive 
\[
\begin{gathered}
| G_c(x,x',\lambda)|\le C (1+\frac1{\sqrt{|x_2-x_2'|}} ), \\
\lim_{|x|\to\infty}G_c(x,x',\lambda)\ast f(x')  \to 0.
\end{gathered}
\]

Via the Sato theory and ingenious ideas, Boiti, Pempinelli, and Pogrebkov derive a decomposition formula for $G_d$ in \cite[(3.15)-(3.17)]{BP211} or \cite[(3.20)-(3.22)]{BP212}. Then applying the totally positive condition $A\in \textrm{Gr}(N, M)_{> 0}$ to that decomposition formula and similar argument as that in \cite[Proposition 2.1]{Wu18},   one can justify
\[
\begin{gathered}
|G_d(x,x',\lambda)|<C,\\
 \lim_{|x|\to\infty}G_d(x,x',\lambda)\ast f(x')  \to 0 .
\end{gathered}
\]

\end{proof}

Thanks to Proposition \ref{P:eigen-green}, especially, the technical reasons in the derivation of analytical properties of $G_d$, we restrict our focus to totally positive $Gr(N, M)_{> 0}$ KP solitons from now on.

\begin{theorem}\label{E:exitence-spectral} Suppose $u(x)=u_0(x)+v_0(x)$ with $u_0$  a totally positive $Gr(N, M)_{> 0}$ KP soliton,   $\partial_x^kv_0 \in  {L^1\cap L^\infty}$, $   |k|\le 2 $, $|v_0|_{{L^1\cap L^\infty}}\ll 1$, and $v_0(x) $ a real-valued function. Then, for fixed $\lambda\in \CC \backslash\{ 0,\kappa_1,\cdots,\kappa_M\}$, there is a   unique solution  $m(x, \lambda)$ to the the boundary value problem \eqref{E:Lax-bdry}.

\end{theorem}
\begin{proof}
Applying {Proposition} \ref{P:eigen-green} and the assumption $\partial_ x^kv_0 \in  {L^1\cap L^\infty}$, $ 0\le |k|\le 2 $, $|v_0|_{{L^1\cap L^\infty}}\ll 1$, for $ \lambda\ne 0$, $\lambda_R\notin  \{   \kappa_1,\cdots, \kappa_M  \}$, one can prove  the unique solvability of an $L^\infty$ solution to \eqref{E:spectral-integral}. 
Moreover, from \eqref{E:sym-0} and Lemma \ref{L:green-heat},  the unique solvability of \eqref{E:Lax-bdry} is equivalent to that of \eqref{E:spectral-integral}. 
\end{proof}

\section{Direct Problem : scattering data} \label{S:sd}

The goal of the inverse scattering theory is to prove that the  eigenfunction $m(x,\lambda)$ can be characterized by special analytical and algebraic constraints (scattering data) and these constraints will evolve with time linearly. In this section, we shall construct  the forward scattering transformation which maps the potential $u(x)$ to these scattering data and prove that it can linearise the  Kadomtsev-Petviashvili equation \eqref{E:KPII-intro}.  

To start, we present analytical and algebraic constraints of the Green function.
\begin{lemma}\label{L:discontinuities}
The Green function $ G$,   associated to a totally positive $Gr(N, M)_{> 0}$ KP soliton $u_0(x)$, 
satisfies 
\begin{align}
(i)\ &G(x,x',\lambda)=\overline{G(x,x',\overline\lambda)};\label{E:green-sym}\\
(ii)\ &\textit{near the pole of the Sato normalized  eigenfunction $\chi $, namely, $\lambda\in D_{0}^\times$,}\label{E:G-expansion}\\
&G(x,x',\lambda)=\sum_{m=0}^{N-1}G_0^{(m)}(x,x') \lambda ^m+\omega_0(x,x',\lambda),\nonumber\\
&|\lambda^m G_0^{(m)}|_{L^\infty(D_0)},\, |\frac{G_0^{(m)}}{ 1+|x-x'|^m }|_{L^\infty(D_0)}  ,\,
  |\frac{\omega_0}{|\lambda|^{n-k}\,(1+|x-x'|^{n-k})}|_{L^\infty(D_0)}\nonumber\\
&  \le C (1+\frac1{\sqrt{|x_2-x_2'|}} ),\qquad 0\le m\le N-1,\ 0\le k\le n\le N;\nonumber\\
(iv)\ &\textit{near the poles of the Sato normalized adjoint eigenfunction $\xi $, i.e., $\lambda\in D_{ \kappa_j}^\times$,}\label{E:g-asym-pm-i-prep-limit-sym}\\
&G(x,x', \lambda) 
=     
{ \mathfrak G}_j(x,x')+\frac {1}\pi \chi_j(x)\xi_j(x')
\cot^{-1}\frac{\lambda_R-\kappa_j}{|\lambda_I|}+\omega_j(x,x', \lambda) , \nonumber\\
&\cot^{-1}\frac {\lambda_R-\kappa_j}{|\lambda_{I}|}
= \left\{
{\begin{array}{lcl}
   \alpha,&0<\alpha\le\pi , & \lambda\in D_{\kappa_j}^\times,\\
2\pi-\alpha,&\pi\le\alpha<2\pi,& \lambda\in D_{\kappa_j}^\times,
\end{array}}
\right. \nonumber\\
&|{\mathfrak G}_j  |_{L^\infty(D_{ {\kappa}_j})},\,  |\omega_j  |_{L^\infty(D_{ {\kappa}_j})} ,\,
 |\frac {\omega_j (x,x',\lambda)}{(\lambda-\kappa_j)(1+|x' -x |)} |_{L^\infty(D_{ {\kappa}_j})}\nonumber  \\
 &\le C(1+\frac1{\sqrt{|x_2-x_2'|}} ),\qquad 1\le j\le M.\nonumber
\end{align}
Moreover,  then the Green function $\mathcal G$ satisfies the symmetry (under the mod $M$-condition) \cite{BP214} 
\beq\label{E:g-3-x-fix-theorem}
\mathcal G (x,x',\kappa^+_{j-1})=\mathcal G (x,x',\kappa_j^+)+\varphi_j(x)\psi_j(x'),\ \kappa^+_j=\kappa_j+0^+. 
\eeq
\end{lemma}

\begin{proof} We omit the details for the proof of  \eqref{E:green-sym}-\eqref{E:g-asym-pm-i-prep-limit-sym}  and refer \cite[Lemma 2.3, Theorem 1]{Wu18} for a similar detailed proof.

 To prove   \eqref{E:g-3-x-fix-theorem}, one can follow the argument in \cite[ \S 3]{BP214} which relies on the condition  $|x-x'|$ is bounded.   Precisely, from \eqref{E:sato}, \eqref{E:pole}, \eqref{E:residue-eigenfunction}, Lemma \ref{L:green-heat}, we decompose 
\begin{align}
    \mathcal G (x,x',\lambda) = &g_M(x,x',\lambda)+f(x,x',\lambda),\label{E:G-sym-+}\\   
 g_M(x,x',\lambda) 
=&-\frac{\theta (x_2-x_2')}{2\pi}[\int_\RR  e^{(\lambda_R+is)(x_1-x_1')+(\lambda_R+is)^2(x_2-x_2') }( 1+ \sum_{j=1}^M\frac{\chi_j (x)\xi_j(x')}{is+\lambda_R-\kappa_j} ) ds\nonumber\\
-&  2\pi \sum_{j=1}^M\theta(\lambda_R-\kappa_j)\varphi_j(x)\psi_j(x')   ],\nonumber\\
 f(x,x',\lambda) 
=& \frac{ 1}{2\pi}\int_{-|\lambda_{ I}|}^{|\lambda_{ I}|}  e^{(\lambda_R+is)(x_1-x_1')+(\lambda_R+is)^2(x_2-x_2') }( 1+ \sum_{j=1}^M\frac{\chi_j (x)\xi_j(x')}{is+\lambda_R-\kappa_j} ) ds\nonumber\\
-&  \sum_{j=1}^M\theta(\lambda_R-\kappa_j)\varphi_j(x)\psi_j(x') .  \nonumber
\end{align} If $|x-x'|$ is bounded, applying the residue theorem, 
\begin{align}
 &g_M(x,x',\lambda)\label{E:G-sym-+-1} \\
=&-\frac{\theta (x_2-x_2')}{2\pi}[ \int_\RR  e^{(\lambda_R+is)(x_1-x_1')+(\lambda_R+is)^2(x_2-x_2') }ds 
\nonumber\\
+& \sum_{j=1}^M\chi_j (x)\xi_j(x') \int_\RR   \frac{e^{(\lambda_R+is)(x_1-x_1')+(\lambda_R+is)^2(x_2-x_2') }-e^{ \kappa_j (x_1-x_1')+\kappa_j^2(x_2-x_2') }}{is+\lambda_R-\kappa_j}  ds \nonumber\\
+&\sum_{j=1}^M\varphi_j (x)\psi_j(x')\int_\RR   \frac{1}{is+\lambda_R-\kappa_j}  ds  -2\pi \sum_{j=1}^M\theta(\lambda_R-\kappa_j)\varphi_j(x)\psi_j(x')   ]\nonumber\\
=&-\frac{\theta (x_2-x_2')}{2\pi}[ \int_\RR  e^{(\lambda_R+is)(x_1-x_1')+(\lambda_R+is)^2(x_2-x_2') }ds\nonumber\\
+& \sum_{j=1}^M\chi_j (x)\xi_j(x') \int_\RR   \frac{e^{(\lambda_R+is)(x_1-x_1')+(\lambda_R+is)^2(x_2-x_2') }-e^{ \kappa_j (x_1-x_1')+\kappa_j^2(x_2-x_2') }}{is+\lambda_R-\kappa_j}  ds]\nonumber\\
\equiv &g_M(x,x').\nonumber
\end{align}
Using a similar argument, Lemma \ref{L:reside},  and  
\[
\int_{-|\lambda_I|}^{|\lambda_I|}\frac{1}{s-i(\lambda_R-\kappa_j)}ds
= 
2 \pi i[\theta(\lambda_R-\kappa_j)-1] +2i\cot^{-1}\frac{\lambda_R- \kappa_j}{ |\lambda_I|}, \ \lambda\in D_{\kappa_j}^\times , 
\] one obtains
\begin{align}
&f(x,x',\lambda)\label{E:G-sym-+-2}\\
=& \frac{ 1}{2\pi}\int_{-|\lambda_{ I}|}^{|\lambda_{ I}|}  e^{(\lambda_R+is)(x_1-x_1')+(\lambda_R+is)^2(x_2-x_2') }ds\nonumber\\
 +&\frac{  1}{2\pi}\sum_{j=1}^M\chi_j (x)\xi_j(x')\int_{-|\lambda_{ I}|}^{|\lambda_{ I}|}  \frac{ e^{(\lambda_R+is)(x_1-x_1')+(\lambda_R+is)^2(x_2-x_2') }- e^{\kappa_j(x_1-x_1')+\kappa_j^2(x_2-x_2') }}{is+\lambda_R-\kappa_j}  ds\nonumber\\
 +&\frac{ 1}{2\pi} \sum_{j=1}^M\varphi_j (x)\psi_j(x')\{2 \pi  [\theta(\lambda_R-\kappa_j)-1] +2 \cot^{-1}\frac{\lambda_R- \kappa_j}{ |\lambda_I|} \}\nonumber\\
   -& \sum_{j=1}^M\theta(\lambda_R-\kappa_j)\varphi_j(x)\psi_j(x') .
 \nonumber
\end{align}
Combining \eqref{E:G-sym-+} - \eqref{E:G-sym-+-2}, and applying Lemma \ref{L:reside}, if $|x-x'|$ is bounded, we obtain
\begin{align}
&\mathcal G (x,x',\lambda)\label{E:g-3-x-fix} \\
= &g_M(x,x')+
\frac{ 1}{2\pi}\int_{-|\lambda_{ I}|}^{|\lambda_{ I}|}  e^{(\lambda_R+is)(x_1-x_1')+(\lambda_R+is)^2(x_2-x_2') }ds\nonumber\\ 
+&\frac{  1}{2\pi}\sum_{j=1}^M\chi_j (x)\xi_j(x') \int_{-|\lambda_{ I}|}^{|\lambda_{ I}|}  \frac{ e^{(\lambda_R+is)(x_1-x_1')+(\lambda_R+is)^2(x_2-x_2') }- e^{\kappa_j(x_1-x_1')+\kappa_j^2(x_2-x_2') }}{is+\lambda_R-\kappa_j}  ds\nonumber\\
+&\frac 1\pi \sum_{j=1}^M \varphi_j(x)\psi_j(x')\cot^{-1}\frac{\lambda_R- \kappa_j}{ |\lambda_I|}.
\nonumber
\end{align}
 Consequently, using $|x-x'|<\infty$,
\beq\label{E:g-x-fixed-sym}
\mathcal G (x,x',\lambda) = g_M(x,x')+  \sum_{j=1}^M \varphi_j(x)\psi_j(x')\theta(\kappa_j-\lambda),\  \forall \lambda\in\RR.
\eeq So \eqref{E:g-3-x-fix-theorem} is verified.

\end{proof}

The following lemma will be used to show that $m(x,\lambda)$ satisfies a $\overline\partial$ problem.
\begin{lemma} \label{L:cont-debar} \cite{BP214}, \cite{VA04}
For $\lambda_I\ne 0$,
\[
\partial_{\bar\lambda}  G(x,x' ,\lambda)=  
 -\frac {\sgn(\lambda_I)}{2\pi i}e^{(\overline\lambda-\lambda)(x_1-x_1')+(\overline\lambda ^2-\lambda^2)(x_2-x_2')} \chi(x,\overline\lambda )\xi(x',\overline\lambda).
\] 

\end{lemma}

\begin{proof}  Via the same argument as that in the proof in \cite[Lemma 2.4]{Wu19}, one obtains 
\begin{align}
 \partial_{\overline\lambda}  G _c(x,x', \lambda)
=&-\frac {\mbox{sgn}(\lambda_I)}{2\pi i}e^{[\overline\lambda-\lambda](x_1-x_1')+[\overline\lambda ^2-\lambda^2](x_2-x_2')} \chi(x,\overline\lambda )\xi(x',\overline\lambda)\label{E:pm}\\
+&\frac 1{2 }\theta(x'_2-x_2)\sum_{j=1}^M     e^{-\lambda(x_1-x_1') -\lambda ^2 (x_2-x_2')} \varphi_j(x)\psi_j(x')\delta(\lambda_R-\kappa_j)\nonumber ,\nonumber\\
 \partial_{\overline\lambda}  G _d(x,x', \lambda)
=\ & -\frac 12\theta(x_2'-x_2)\sum_{j=1}^M  {e^{ -\lambda (x_1-x_1')-\lambda^2 (x_2-x_2') }}  
\varphi_j(x)\psi_j(x')\delta(\lambda_R-\kappa_j) .\label{E:dbar-g-d} 
\end{align}Hence follows the lemma.
\end{proof}

Based on the characterization of the Green function $G$, we can provide analytic constraints, algebraic constraints, and  the $\overline\partial$ data of $m$ in Theorem \ref{T:KP-eigen-existence} and \ref{T:sd-continuous}. 

\begin{theorem} \label{T:KP-eigen-existence}  
 Let $u(x)=u_0(x)+v_0(x)$, $u_0(x)$ a totally positive $Gr(N, M)_{> 0}$ KP soliton, $v_0(x)$ a
real-valued function  with $\partial^k_xv_0\in L^1\cap L^\infty$ for $|k|\le 2$ and $|v_0|_{{L^1\cap L^\infty}}\ll 1$. Then the eigenfunction  $m(x, \lambda)$ derived from Theorem  \ref{E:exitence-spectral}  satisfies  
\begin{align}
(i)\ &  m(x,\lambda)=\overline{m(x,\overline\lambda)};\label{E;m-reality}\\
(ii)\ &   { |(1- E_{0 }  )  m |_{L^\infty}\le C} ;\label{E:bdd}\\
(iii)\ &  \textit{near the pole, namely, $\lambda\in D_0^\times$,}\label{E:pole-m} \\
& m(x,\lambda)=  { \sum_{n=1}^N\frac{ m_{\res,n}(x )}{\lambda ^n} } +  m_{0,r}(x, \lambda),
  \ \ 
  m_{\res,n}(x )\in \RR,\nonumber\\
& | \frac{ m_{\res,n}}{1+|x|^{N-n}}|_{L^\infty {(D_{0})}},\,
|\frac{ m_{0,r} }{1+|x|^{N}} |_{L^\infty {(D_{0})}} \le C  |v_0|_{L^1\cap L^\infty} , \ 1\le n\le N  ;\nonumber\\
(iv)\ &\textit{near the discontinuities $\kappa_j$, namely, $\lambda\in D_{\kappa_j}^\times$,} \label{E:discon}\\
&m(x,\kappa_j+0^+e^{i\alpha})\equiv m(x,\lambda_0)=\frac{\Theta_j(x)}{1-\gamma_j \cot^{-1}\frac{\lambda_{0,R}-\kappa_j}{|\lambda_{0,I}|} }  ,\nonumber\\  
& |  m   |_{L^\infty {(D_{\kappa_j})}},\   |\frac{\frac\partial{\partial s}  m (x,\lambda) }{1+|x |}|_{L^\infty {(D_{\kappa_j})}}\le C |v_0|_{L^1\cap L^\infty}, \ 1\le j\le M,
\nonumber\\
& \Theta_j(x)=  [1+{\mathfrak G}_j(x,x')\ast  v_0(x')]^{-1}\chi_j(x'),\  
 \gamma_j =  -\frac 1\pi\iint\xi_j(x )v_0(x )\Theta_j(x)dx.\nonumber
\end{align}
Moreover, the Jost function $\Phi$ satisfies the algebraic constraints  \cite{BP214}
\beq\label{E:m-sym}
\Phi(x,\kappa)  {\mathcal D}^\flat=0 ,
\eeq where
\begin{align}
&\Phi(x,\kappa)= (\Phi_1(x), \cdots, \Phi_M (x)),\,
\Phi_j  (x)=  e^{\kappa _jx_1+ \kappa _j ^2x_2}m  (x, \kappa^+_j)  , \,\kappa^+_j=\kappa_j+0^+,\label{E:m-sym-cond}\\
 & 
  {\mathcal D}_{ji}^\flat= \mathcal D_{ji}+\sum_{l=j}^M\frac{c_{jl}\mathcal D_{li}}{ 1-c_j} ,\ 
 1\le  j\le M,\ \ 1\le i\le N\nonumber\\
& c_{jl}=-\int v_0(x)\Psi _j(x)\varphi_l(x )dx,\ c_j=c_{jj},\,
\Psi _j(x)=res_{\lambda\to\kappa_j^+}\Psi (x,\lambda) ,\nonumber\\
&\mathcal L^\dagger\Psi (x,\lambda)\equiv \left(\partial_{x_2}+\partial^2_{x_1}+u_0(x)\right)\Psi (x,\lambda)=-v_0(x)\Psi (x,\lambda)  \nonumber
\end{align}and $\mathcal D$ defined by Definition \ref{D:dd'}.

\end{theorem}
\begin{proof} $\underline{\emph{Step 1  (Proof of \eqref{E;m-reality}-\eqref{E:pole-m})}}:$ The reality condition \eqref{E;m-reality} follows from \eqref{E:spectral-integral}, \eqref{E:green-sym}  and $v_0(x)$ is real-valued.  
   Applying \eqref{E:spectral-integral} and Proposition \ref{P:eigen-green},
\[
|(1-E_0)m(x,\lambda)|=|(1+G\ast v_0)^{-1}(1-E_0)\chi|\le C .
\] So \eqref{E:bdd} is justified. 
Besides,  for $\lambda\in D_0$, \eqref{E:G-expansion} and the resolvent identity imply, there exist operators $P^{(k)}_0(x)$, $P^{(N)}_0(x,\lambda)$,  $0\le k \le N-1$ ,  
\beq\label{E:resolvent-1}
\begin{gathered}
P^{(k)}_0f=\int _{\RR}\wp^{(k)}_0(x,x')f(x')dx',\\
|\frac{\wp^{(k)}_0}{ 1+|x-x'|^k }|,\ |\frac{\wp^{(N)}_0}{\lambda^N( 1+|x-x'|^{N})}|  \le C (1+\frac1{\sqrt{|x_2-x_2'|}} ), 
\end{gathered}\eeq such that
\begin{align}
 & (1+G \ast v_0)^{-1}\label{E:resolvent}\\
 =&(1+[\sum_{k=0}^{N-1}\lambda^k G_0^{(k)}+\omega_0]\ast v_0)^{-1}\nonumber\\
 =&\sum_{k=0}^{N-1} (1+G^{(0)}_0	\ast v_0)^{-1} \left[\Delta G_0^{(0)} \ast v_0 (1+G^{(0)}_0	\ast v_0)^{-1}\right]^k\nonumber\\
 +& \left[\Delta G_0^{(0)} \ast v_0 (1+G^{(0)}_0	\ast v_0)^{-1}\right]^N(1+G \ast v_0)^{-1}\nonumber\\
=&\sum_{k=0}^{N-1}P^{(k)}_0(x) \lambda^k +P^{(N)}_0 (x,\lambda). \nonumber
\end{align}
From \eqref{E:sato}, one can set $\chi(x,\lambda)=\sum_{m=0}^N\dfrac{ \chi_{res,m}(x )}{\lambda ^m}$ for $\chi_{res,m}(x)\in L^\infty$. Along with \eqref{E:resolvent-1}, \eqref{E:resolvent}, and  Proposition \ref{P:eigen-green},
\begin{align*}
 m(x,\lambda) 
=  \left(\sum_{k=0}^{N-1 }P^{(k)}_0 (x) \lambda^k  +P^{(N)}_0 (x,\lambda)\right)  \sum_{m=0}^N\frac{ \chi_{res,m}(x )}{\lambda ^m} ,
\end{align*}  which yields
\beq\label{E:g-decomp}
\begin{split}
 m_{res,n}(x)= &\sum_{l=0}^{N-n}P_0^{(l)} \chi_{res,n+l}(x),\ 1\le n\le N,\\
 m(x,\lambda)= & { \sum_{n=1}^N\frac{ m_{res,n}(x )}{\lambda ^n} } +  m_{0,r}(x, \lambda).
\end{split}
\eeq 
So \eqref{E:pole-m} follows from \eqref{E:resolvent-1} and \eqref{E:g-decomp}.

$\underline{\emph{Step 2  (Proof of \eqref{E:discon} - \eqref{E:m-sym-cond})}}:$ Proof of \eqref{E:discon} can be derived similarly as that in the proof of \cite[Theorem 1, Step 3]{Wu19}. 

The $\mathcal D^\flat$-symmetry  \eqref{E:m-sym} is  crucial for defining the discrete data (norming constants) for perturbed $\mathrm{Gr}(N, M)_{\ge 0}$ solitons when $M>2$. For convenience, we sketch the  proof and refer details to \cite[\S 4]{BP214}.  

First of all,   Proposition \ref{P:eigen-green} implies there exists a unique $\mathcal K(x,x',\lambda)$  satisfying the following integral equations,
\beq\label{E:total-green}
\begin{split}
\mathcal K(x,x',\lambda)=&\mathcal G(x,x',\lambda) -\mathcal G(x,y,\lambda)\ast _{y} v_0(y)\mathcal K(y,x',\lambda)\\
\equiv&\mathcal G-\mathcal G\ast v_0\,\mathcal K,\\
\mathcal K(x,x',\lambda)=&\mathcal G(x,x',\lambda)- \mathcal K (x,y,\lambda)\ast_{y} v_0(y)\mathcal G(y,x',\lambda)\\
\equiv&\mathcal G-\mathcal K\ast v_0\,\mathcal G.
\end{split}
\eeq The function $\mathcal K(x,x',\lambda)$ is called the total Green function since
\beq\label{E:sym-0-K}
\begin{gathered}
\overrightarrow{\mathcal L_{v_0}}\mathcal K  =\mathcal K\overleftarrow{\mathcal L_{v_0}}=\delta(x-x') , 
 \\ 
\mathcal L_{v_0}=\mathcal L+v_0,\ \mathcal L=-\partial_{x_2}+\partial^2_{x_1}+u_0(x).
\end{gathered}
\eeq Here $\overrightarrow {\mathcal L}$ denotes the operator $\mathcal L$ applying to the $x$ variable of $\mathcal K$ and $\overleftarrow{\mathcal L}$ denotes the operator  applying to the $x'$ variable of $\mathcal K$. Moreover, the eigenfunction $\Phi(x ,\lambda)$ and adjoint eigenfunction $\Psi(x ,\lambda)$ defined by \eqref{E:Lax-bdry} and \eqref{E:m-sym-cond} satisfy
\beq\label{E:eigen-adj-eigen}
\begin{split}
\Phi(x,\lambda)=& \mathcal K(x,x',\lambda)\ast_{x'}\overleftarrow{\mathcal L}\varphi (x',\lambda)
\equiv  \mathcal K \ast \overleftarrow{\mathcal L}\varphi ,\\ 
\Psi(x' ,\lambda)=&\psi(x ,\lambda)\ast_{x }\overrightarrow{\mathcal L}\mathcal K (x,x',\lambda)
\equiv  \psi\ast\overrightarrow{\mathcal L}\mathcal K  .
\end{split}
\eeq Taking $\lambda$ limits of \eqref{E:total-green} and \eqref{E:eigen-adj-eigen}   yield
\beq\label{E:eigen-green-lim}
\begin{gathered}
\mathcal K_j=\mathcal G_j-\mathcal G_j\ast v_0\,\mathcal K_j,\\ 
\Phi_j=\mathcal K_j\ast \overleftarrow{\mathcal L}\varphi_j,\quad \Psi _j=\psi_j \ast \overrightarrow{\mathcal L}\mathcal K_j,
\end{gathered}
\eeq where we use the convention $\mathcal G_j=\mathcal G(x,x',\kappa_j^+)$, $\mathcal K_j=\mathcal K(x,x',\kappa_j^+)$, and $\Phi_j=\Phi(x,\kappa_j^+)$. 

Using \eqref{E:total-green}-\eqref{E:eigen-green-lim}, \eqref{E:g-3-x-fix-theorem}, and
\begin{align*}
 c_{jl}=&-\int v_0(x)\Psi _j(x)\varphi_l(x )dx=\int \psi_j(x)\overleftarrow{\mathcal L} \varphi_l(x )dx= \Psi_j\ast  \overleftarrow{\mathcal L}\varphi_l,
\end{align*}one can derive
\beq\label{E:K-3-x-fix-theorem} 
\mathcal K_{j-1}=\mathcal K_j+\frac{\Phi_j(x)\Psi_j(x')}{1-c_j},
\eeq which implies
\beq\label{E:K-general-sum} 
\sum_{j= 1}^{ M}\frac{\Phi_j(x)\Psi_j(x')}{1-c_j}=0.
\eeq and
\beq\label{E:K-general} 
\mathcal K_{l}=\mathcal K_i+\sum_{j=l+1}^{i+M}\frac{\Phi_j(x)\Psi_j(x')}{1-c_j}.
\eeq Here $
c_j= c_{jj}$ and the mod $M$ - condition is adopted. 

Applying $  \overleftarrow{\mathcal L}\varphi_i$ to \eqref{E:K-general} from the right and using \eqref{E:eigen-green-lim}, we obtain
\beq\label{E:K-phi-general} 
\mathcal K_l \ast \overleftarrow{\mathcal L}\varphi_i=\Phi_i+\sum_{j=l+1}^{i+M}\frac{\Phi_j(x)c_{ji}}{1-c_j}.
\eeq  Summing \eqref{E:K-phi-general} up with the matrix $\mathcal D_{im}$ and using the symmetry \eqref{E:D-sym}, we derive
\beq\label{E:higher-kdv-sym}
\sum_{i=1}^{M} \Phi_i\mathcal D_{im}+\sum_{i=1}^{M}\sum_{j=l+1}^{i+M}\frac{\Phi_j(x)c_{ji}\mathcal D_{im}}{1-c_j}=0.
\eeq Taking $l=M$ in \eqref{E:higher-kdv-sym} and using \eqref{E:K-general-sum},   we obtain \eqref{E:m-sym}.

\end{proof}

We emphasize that the remarkable $\mathcal D^\flat$-symmetry \eqref{E:m-sym}, 
introduced by Boiti, Pempenelli, and Pogrebkov \cite{BP214}, depends not only the data $\kappa_1<\cdots<\kappa_M$,  
 $A=(a_{ij})\in\mathrm{Gr}(N, M)_{\ge 0}$, but also on the perturbation $v_0(x)$. In case $(N,M)=(1,2)$, it defines a different  but equivalent symmetry from the KdV symmetry \cite[Theorem 2, Theorem 4, (3.5), (3.7), (3.40)]{Wu18}, \cite[(3.4)]{Wu19}. 

\begin{theorem}\label{T:sd-continuous}   

 Let $u(x)=u_0(x)+v_0(x)$, $u_0(x)$ a totally positive $Gr(N, M)_{> 0}$ KP soliton, $v_0(x)$ a
real-valued function  with $\partial^k_xv_0\in L^1\cap L^\infty$ for $|k|\le 2$ and $|v_0|_{{L^1\cap L^\infty}}\ll 1$. Then the solution  $m(x, \lambda)$ obtained from Theorem \ref{E:exitence-spectral}  satisfies
\beq\label{E:conti}
\partial_{\overline\lambda}m(x, \lambda)
=  s_c(\lambda) e^{(\overline\lambda-\lambda)x_1+(\overline\lambda^2-\lambda^2)x_2  }m{(x, \overline\lambda)} , \ \lambda_I\ne 0,
\eeq with
\beq\label{E:conti-sd}
\begin{gathered}
  s_c(\lambda) = \frac {\sgn(\lambda_I)}{2\pi i} \iint e^{-[(\overline\lambda-\lambda)x_1+(\overline\lambda^2-\lambda^2)x_2 ] } \xi(x, \overline\lambda) v_0(x)m(x, \lambda)dx ,\\
s_c(\lambda)=\overline{s(\overline\lambda)}   .
\end{gathered}
\eeq  
 Moreover, 
\beq
| (1-E_{\cup_{1\le j\le M} D_{\kappa_j}}  )s_c  |_{   L^2(|\lambda_I| d\overline\lambda \wedge d\lambda)\cap L^\infty}  \le  C\sum_{|k|\le 2} |\partial_x^kv_0| _{L^1\cap L^\infty},
\label{E:pm-i-new}  
\eeq
and 
\begin{gather} 
  s_c(\lambda)=
\left\{
{\ba{ll}
 \frac{ \frac {i}{ 2} \sgn(\lambda_I)}{\overline\lambda-\kappa_j}\frac {\gamma_j}{1-\gamma_j { \cot^{-1}\frac { \lambda_R-\kappa_j}{|\lambda_I|}}}+\sgn(\lambda_I)h_j(\lambda),&\lambda\in D^ \times_{ \kappa_j },\\
\sgn(\lambda_I) {\hbar_0}(\lambda),&\lambda\in  D^\times _{ 0},
\ea}
\right.
\label{E:cd-decomposition} 
\end{gather}
for $1\le j\le M$, with
\beq \label{E:cd-decomposition-new}
\begin{gathered}
 |\gamma_j|_{L^\infty}\le |v_0|_{L^1}, \ 
       \sum_{0\le |l|\le 1}(| \partial_{\lambda_R}^{l_1}\partial_{\lambda_I}^{l_2}h_j|_{L^\infty}+|\partial_{\lambda_R}^{l_1} \partial_{\lambda_I}^{l_2}\hbar_0|_{L^\infty})\le C |(1+|x|^2)     v_0|_{L^1\cap L^\infty},   \\      
 h_j(\lambda)=-\overline{h_j( \overline\lambda)}, \ 
 \hbar_0(\lambda)=-\overline{\hbar_0( \overline\lambda)} . 
\end{gathered}
\eeq  
  
\end{theorem}

\begin{proof}  The theorem can be proved by  the same approach as that in \cite[Theorem 2, (3.10)-(3.12)]{Wu19}.
\end{proof}

One can adapt the approach as that in \cite[Theorem 3, Theorem 4]{Wu19} to derive the following Cauchy integral equation for $m(x,\lambda)$,
\beq\label{E:cauchy-integral-equation-2-pole}
\begin{array}{cl}
  {  m}(x, \lambda) =1+\sum_{n=1}^N\frac{  m_{ {\res},n }(x  ) }{\lambda^n  }  +\mathcal CT
   m , & \forall\lambda\ne  0.
\end{array}
\eeq 
Nevertheless, if we use the above Cauchy integral equation to solve the inverse problem, then the multiple pole property at $0$ could form an obstruction.   Since it yields a singular integral with a kernel blowing up of order $N+1$ at $0$ and causes troubles when $N\ge 2$. To remedy the situation, we introduce 
\begin{definition}\label{D:reg-m} A regularization eigenfunction is defined by
\beq\label{E:regularized-m}
\begin{gathered}
\widetilde m (x, \lambda)=\frac {(\lambda-z_1)^{N-1}}{\Pi_{2\le n\le N}(\lambda- z_n)}m(x ,\lambda),\\  
 z_1=0,\, z_n=\frac{n-1}N  a,\,2\le n\le N,\\
 \tilde a= a/N, \ \widetilde D_{z} =D_{z,\tilde  a},\ \widetilde  D^ \times_{z}=\widetilde D_{z}/\{z\},
\end{gathered}
\eeq where $ a$ and $D_{z,r}$ are defined by Definition \ref{D:terminology}.
\end{definition}

\begin{theorem}\label{T:normal-eigen} If $u(x)=u_0(x)+v_0(x)$, $u_0(x)$ is a totally positive $Gr(N, M)_{> 0}$ KP soliton, $v_0(x)$ a real-valued function with $\partial_x^kv_0 \in  {L^1\cap L^\infty}$, $   |k|\le 2 $, $|v_0|_{{L^1\cap L^\infty}}\ll 1$, then for fixed $\lambda\in \CC \backslash\{ z_n,\kappa_j\}$, $1\le n\le N,\ 1\le j\le M$,   there is a   unique solution  $\widetilde m(x, \lambda)$ to the spectral equation
\beq\label{E:reg-direct-spectral}
\left\{
{\ba{ll}
 L\widetilde m(x,\lambda)=- v_0(x)\widetilde m(x,\lambda),& \\
\lim_{|x|\to\infty}  ( \widetilde  m(x,\lambda)-\frac {(\lambda-z_1)^{N-1}}{\Pi_{2\le n\le N}(\lambda- z_n)}\chi(x,\lambda))=0, 
\ea}
\right.
\eeq  and    
\begin{align}
(i)\ &\widetilde m(x, \lambda)=\overline{\widetilde m(x, \overline\lambda)};\label{E:tilde-reality}\\
(ii)\ &  |(1- E_{\cup_{1\le n\le N} \widetilde D_{z_n}  } \widetilde m(x, \lambda)|\le C; \label{E:define-remainder-m-sf-pm2-0}\\
(iii)\ &\textit{near the poles $z_n$, namely, $\lambda\in\widetilde  D^ \times_{z_n}$, $ 1\le n\le N$,}\label{E:define-remainder-m-sf-pm2} \\
&\widetilde m (x, \lambda)  
 =    \frac{\widetilde m_{z_n,\res}(x)}{\lambda-z_n}  +\widetilde m_{z_n,r}(x, \lambda),
 \ 
\widetilde m_{z_n,\res}(x)\in  \RR, \nonumber\\
&|\widetilde m_{z_n,\res}(x)|_{L^\infty}, \ 
|(\lambda-z_n)\widetilde m_{z_n,r}(x, \lambda) |_{L^\infty(\widetilde D_{z_n})},\ | \frac{\widetilde m_{z_n,r}(x, \lambda)}{ 1+|x| } |_{L^\infty(\widetilde D_{z_n})}\nonumber\\
&\hskip3in\le  C |v_0|_{L^1\cap L^\infty} ;   
\nonumber\\
(iv)\ &\textit{near the discontinuities, namely, $\lambda\in\widetilde  D^ \times_{\kappa_j}$, $ 1\le j\le M$,}\label{E:define-remainder-m-sf-new}\\
&{\widetilde m}(x, \kappa_j+0^+e^{i\alpha}) 
\equiv {\widetilde m}(x, \lambda_0)  
 =  \frac{  \frac{(\kappa_j-z_1)^{N-1}}{\Pi_{2\le n\le N}(\kappa_j-z_n)}\Theta_j(x)}{1-\gamma_j\cot^{-1}\frac {\lambda_{0,R}-\kappa_j}{|\lambda_{0,I|}}}, \nonumber\\
&|\widetilde m   |_{L^\infty(\widetilde D_{\kappa_j})},\,
|\frac{\frac\partial{\partial s} \widetilde m (x, \lambda)}{1+|x| } |_{L^\infty(\widetilde D_{\kappa_j})}\le C |v_0|_{L^1\cap L^\infty},\,  
 \textit{$\Theta_j$, $\gamma_j$ defined by \eqref{E:discon}. }
 \nonumber
    \end{align}
 Besides, $ (\widetilde\Phi_1(x), \cdots, \widetilde\Phi_M  (x))  {\mathcal D}^\sharp =0$ with  
 \beq\label{E:sharp}
 \begin{gathered}
 \widetilde\Phi_j  (x)=e^{\kappa _jx_1+ \kappa _j ^2x_2}\widetilde m  (x, \kappa^+_j) ,\\
 {\mathcal D}^\sharp=\textit{diag}\,(\frac{\Pi_{2\le n\le N}(\kappa_1-z_n)}{(\kappa_1-z_1)^{N-1}}, \cdots, \frac{\Pi_{2\le n\le N}(\kappa_M-z_n)}{(\kappa_M-z_1)^{N-1}})\mathcal D^\flat.
 \end{gathered}\eeq 
So  there exists uniquely  $b\in GL(N\times N)$ such that 
\begin{gather}
\mathcal  D^\sharp  \times b=  \widetilde {\mathcal  D }=\left( \ba{ccc}
{ \kappa_1^N}&{ \cdots}&{ 0 }\\
{ \vdots}&{ \ddots} &{ \vdots}\\
{ 0}&{ \cdots} & \kappa_N^N \\
\widetilde {\mathcal D}_{N+1,1}&\cdots &\widetilde {\mathcal D}_{N+1,N}\\
\vdots &\ddots &\vdots\\
\widetilde {\mathcal D}_{M,1}&\cdots &\widetilde {\mathcal D}_{M,N}
\ea\right), 
\label{E:reg-m-sym}\\
 (\widetilde\Phi_1(x), \cdots, \widetilde\Phi_M  (x))\widetilde{\mathcal D} =0.
\label{E:reg-m-sym-1}
\end{gather}

Moreover,  
\beq\label{E:dbar-m-sf}
\begin{gathered}
  \partial_{\overline\lambda}{\widetilde m}(x,y,\lambda) 
=  \widetilde {  s}_c(\lambda)e^{(\overline\lambda-\lambda)x_1+(\overline\lambda^2-\lambda^2)x_2}\widetilde m(x, \overline\lambda),\ \lambda_I\ne 0,\\
 {\widetilde s}_c(\lambda)=\frac {(\lambda-z_1)^{N-1}\Pi_{2\le n\le N}(\overline\lambda-z_n)}{(  \overline\lambda-z_1)^{N-1}\Pi_{2\le n\le N}(\lambda-z_n)   } 
   s_c(\lambda),\quad \widetilde s_c(\lambda)=\overline{\widetilde s(\overline\lambda)}
\end{gathered}
\eeq with
\beq\label{E:tilde-dbar-m-sf-new-0}
\begin{gathered}
{|{(1-E_{ \cup_{1\le j\le M}\widetilde D_{\kappa_j}}  ) {\widetilde s}_c(\lambda)}|_{   L^2(|\lambda_I| d\overline\lambda \wedge d\lambda)\cap L^\infty} }\le  C\sum_{|k|\le 2}|\partial_x^ kv_0|_{ L^1\cap L^\infty},
\\
 \widetilde s_c(\lambda)=
\left\{
{\ba{ll}
 \frac{ \frac {i}{ 2} \sgn(\lambda_I)}{\overline\lambda-\kappa_j}\frac {\gamma_j}{1-\gamma _j{ \cot^{-1}\frac { \lambda_R-\kappa_j}{|\lambda_I|}}}+\sgn(\lambda_I)\widetilde h_j(\lambda),&\lambda\in \widetilde D^ \times_{ \kappa_j },\ 1\le j\le M,\\
\sgn(\lambda_I) {\widetilde \hbar_n}(\lambda),&\lambda\in  \widetilde D^\times _{ z_n},\ 1\le n\le N,
\ea}
\right.\\
|\gamma_j|_{L^\infty}\le |v_0|_{L^1}, \ 
        \sum_{0\le |l|\le 1} (|\partial_{\lambda_R}^{l_1}\partial_{\lambda_I}^{l_2}\widetilde h_j|_{L^\infty}+  | \partial_{\lambda_R}^{l_1}\partial_{\lambda_I}^{l_2}\widetilde \hbar_n|_{L^\infty})\le C |(1+|x|)     v_0|_{L^1\cap L^\infty},     
 \\
\widetilde h_j(\lambda)=-\overline{\widetilde h_j( \overline\lambda)},\ \widetilde \hbar  _n(\lambda)=-\overline{\widetilde \hbar  _n( \overline\lambda)} .
\end{gathered}\eeq
 
\end{theorem}
\begin{proof} We only give the proof for \eqref{E:define-remainder-m-sf-pm2}. Other properties can be derived directly from Theorem \ref{T:KP-eigen-existence}.

From \eqref{E:spectral-integral} and \eqref{E:regularized-m}, there exist $\chi_{{\res},z_n}(x)\in L^\infty$, $1\le n\le N$,
\beq\label{E:tilde-m-g}
 \begin{gathered}
 \widetilde m(x,\lambda)
= (1+G\ast v_0)^{-1} \frac {(\lambda-z_1)^{N-1}}{\Pi_{2\le n\le N}(\lambda- z_n)}\chi(x,\lambda)\\
\frac {(\lambda-z_1)^{N-1}}{\Pi_{2\le n\le N}(\lambda- z_n)}\chi(x,\lambda)=1+\Pi_{1\le n\le N}\frac {\chi_{{\res},z_n}(x)}{\lambda-z_n}.
\end{gathered}\eeq
Using the argument as in the proof of \eqref{E:G-expansion}, for $\lambda\in \widetilde D_{z_i}\subset D_0$, the Green function satisfies
\beq\label{E:tilde-g-z-n}
\begin{gathered}
G(x,x',\lambda)=G_{z_i}^o(x,x')+\omega_{z_i}(x,x',\lambda),\\
|G_{z_i}^o|, 
| \omega_{z_i} | ,\  |\frac{ \omega_{z_i}}{(\lambda-z_i )(1+|x-x'|)} |\le C (1+\frac1{\sqrt{|x_2-x_2'|}} ).
\end{gathered}
\eeq  Combining \eqref{E:tilde-m-g} and \eqref{E:tilde-g-z-n}, we justify\eqref{E:define-remainder-m-sf-pm2}.
\end{proof}

{Based on the characterization of the eigenfunction $\widetilde m$, we define the eigenfunction space $W$ and the spectral transformation $T$ in Definition \ref{D:quadrature-hat} and \ref{D:spectral}.}

\begin{definition}\label{D:spectral} Given $u(x)=u_0(x)+v_0(x)$, $u_0(x)$   a totally positive $Gr(N, M)_{> 0}$ KP soliton, $v_0(x)$ a real-valued function with $\partial_x^kv_0 \in  {L^1\cap L^\infty}$, $   |k|\le 2 $, $|v_0|_{{L^1\cap L^\infty}}\ll 1$,
define $\mathcal S$ as the {\bf  forward scattering transform} by
\[
\mathcal S(u(x))=\{z_n,\,\kappa_j,\, \widetilde {\mathcal D},\,  \widetilde  s _c(\lambda)\}, 
\]  
where $1\le n\le N $, $1\le j\le M$, 
\begin{align*}
&\textit{$z_n$, location of the simple poles of $\widetilde m$,}\\
&\textit{$\kappa_j$, location of discontinuities of $\widetilde m$,}\\
&\textit{$\widetilde {\mathcal D}$, the norming constants (the $\widetilde {\mathcal D}$-symmetry of $\widetilde \Phi$),}
\end{align*}  defined by \eqref{E:reg-m-sym},  are the {\bf   discrete scattering data};  and  $ \widetilde{  s}_c(\lambda)$ defined by \eqref{E:dbar-m-sf}, is the  {\bf  continuous scattering data}. Denote   $T$ as   the {\bf continuous scattering operator} by
\beq\label{E:cauchy-operator}
T(\phi)(x, \lambda)  =\widetilde {  s}_c(\lambda) e^{(\overline\lambda-\lambda)x_1+(\overline\lambda^2-\lambda^2)x_2  }\phi(x, \overline\lambda).
\eeq 
\end{definition}

So for general $ Gr(N,M)_{> 0}$ KP solitons, the discrete scattering data $\widetilde {\mathcal D}$ defined in Definition \ref{D:spectral} is determined by $\kappa_j$, $z_n$, $A$, as well as the initial perturbation $v_0$. 
This is different from the situation in \cite{Wu19} (or \cite{Wu18}), where for $ Gr(1,2)_{> 0}$ KP solitons, the discrete scattering data is replaced by ${\mathcal D}$,  determined by $\kappa_j$, $z_n$, $A$, and  independent of $v_0$. 

To conclude this section, we justify that for a KPII solution $u(x)=u(x_1,x_2,x_3)$, the scattering data flow $\mathcal S(u)(x_3)$ is linear.
\begin{theorem}\label{T:linearity}    For $\forall x_3\ge 0$, if $\widetilde\Phi(x, \lambda)= \ \widetilde  m(x, \lambda)e^{ \lambda  x_1+ \lambda ^2x_2}$ satisfies the Lax pair \eqref{E:KPII-lax-1} 
with the continuous and discrete scattering data 
\beq\label{E:linear-sd}
\begin{gathered}
\partial_{\overline\lambda}\widetilde  m(x, \lambda)=    {\widetilde   s}_c(\lambda,x_3) e^{(\overline\lambda-\lambda)x_1+(\overline\lambda^2-\lambda^2)x_2}  \widetilde m(x,\overline\lambda)  ,\\
    (\widetilde\Phi(x,\kappa_1^+), \cdots, \widetilde\Phi   (x,\kappa_M^+))\widetilde {\mathcal D} =0. 
\end{gathered} 
\eeq Then  the evolution equations of the scattering data are
\begin{align}
 {\widetilde s}_c(\lambda, x_3)= & {e^{ (\overline\lambda^3-{ \lambda}^3)x_3}}{\widetilde s}_c(\lambda, 0),\label{E:linear-evol-1}\\
 \widetilde {\mathcal {\mathcal D}}_{mn}(x_3)=&{ e^{(\kappa_n^3-\kappa_m^3)x_3}}\widetilde {\mathcal D}_{mn}(0).  \label{E:linear-evol-2}
 \end{align}
\end{theorem}

\begin{proof} We skip the proof of  \eqref{E:linear-evol-1} since it can be proved by the same argument as that in \cite[Lemma 4.1]{Wu19}. On the other hand, from  the $\widetilde {\mathcal D}$ symmetry,
\begin{align}
-\kappa_1^N\widetilde\Phi_1=& \sum_{n=N+1}^M\widetilde {\mathcal D}_{n1}\widetilde\Phi_n, \label{E:M-evol}\\
 \vdots\hskip.3in&\hskip.6in\vdots\nonumber\\
-\kappa_N^N\widetilde\Phi_N=& \sum_{n=N+1}^M\widetilde {\mathcal D}_{nN}\widetilde\Phi_n.\nonumber
\end{align}
  Denote $\mathcal M_\lambda=  -\partial_{x_3}+ \partial_{x_1}^3+\frac 32u\partial_{x_1}+\frac 34u_{x_1}+\frac 34\partial_{x_1}^{-1}u_{x_2}+\tau(\lambda)$ and $\tau(\lambda)= -\lambda^3$. 
So \eqref{E:KPII-lax-1} implies
\begin{align*}
 0=&- \mathcal M_{\kappa_j}(\kappa_j^N\widetilde\Phi_j)=\sum_{n=N+1}^M\mathcal M_{\kappa_j}(\widetilde {\mathcal D}_{nj}\widetilde\Phi_n)\\
 =& \sum_{n=N+1}^M(\widetilde\Phi_n[{ -\partial_{x_3}}+\tau(\kappa_j)]\widetilde {\mathcal D}_{nj} +\widetilde {\mathcal D}_{nj}[{ \mathcal M_{\kappa_j}}-\tau(\kappa_j)]\widetilde\Phi_n)\\
 =& \sum_{n=N+1}^M(\widetilde\Phi_n[-\partial_{x_3}+\tau(\kappa_j)]\widetilde {\mathcal D}_{nj} +\widetilde {\mathcal D}_{nj}[{ \mathcal M_{\kappa_n}-\tau(\kappa_n)}]\widetilde\Phi_n)\\
 =& \sum_{n=N+1}^M \widetilde\Phi_n[{ -\partial_{x_3}-\tau(\kappa_n)+\tau(\kappa_j)}]\widetilde {\mathcal D}_{nj}.
\end{align*} Thus \eqref{E:linear-evol-2} is justified.
\end{proof}

\section{The Cauchy integral equation}\label{S:CIO}

We will derive a Cauchy integral equation for the eigenfunction $\widetilde m(x,\lambda)$. We will first justify a crucial estimate for $\mathcal CT\widetilde m$, an $L^\infty$ estimate for each $x$ fixed.    It is sufficient for deriving a Cauchy integral equation but is insufficient for solving the inverse problem (cf \cite{W87}). Finally, we show a closeness property implied by the Cauchy integral equation and the eigenfunction space characterization (Definition \ref{D:quadrature-hat}).

\begin{definition}\label{D:cauchy}
Let   $\mathcal C$ be the Cauchy integral operator  defined by 
\[
\mathcal C(\phi)(x,  \lambda)=\mathcal C_\lambda(\phi) =-\frac 1{2\pi i}\iint\frac {\phi(x, \zeta)}{\zeta-\lambda}d\overline\zeta\wedge d\zeta. 
\]
 
\end{definition}

We shall apply Liouville's theorem to derive a Cauchy equation. To this aim,  boundedness of $\mathcal CT\widetilde m$ is important which is provided in the following theorem.   
\begin{theorem}\label{T:bdd} Suppose $u(x)=u_0(x)+v_0(x)$, $u_0(x)$   a totally positive $Gr(N, M)_{> 0}$ KP soliton, $v_0(x)$ a real-valued function with $\partial_x^kv_0 \in  {L^1\cap L^\infty}$, $   |k|\le 2 $, $|v_0|_{{L^1\cap L^\infty}}\ll 1$.   Then the eigenfunction $\widetilde m(x,\lambda)$ obtained in Theorem \ref{T:normal-eigen} satisfies
\begin{align*}
&|\mathcal C   T\widetilde m| _{L^\infty} \le C   (1+|x|)  \sum_{|k|\le 2}| \partial_x^k v_0|_{L^1\cap L^\infty}
 ,\\
 & \mathcal C T\widetilde m (x, \lambda)\to 0 ,\quad\textit{ as $|\lambda|\to\infty$, $\lambda_I\ne 0$ .} 
\end{align*}
\end{theorem}

\begin{proof}   One can adapt the argument as that in the proof of \cite[Theorem 3]{Wu19} to prove the theorem. We emphasize that, determined by the heat operator $-\partial_{x_2}+\partial^2_{x_1}+2\lambda\partial_{x_1}+u_0(x)$, the continuous scattering operator $T$ has poles at $\kappa_j$, $1\le j\le M$, and is bounded  in the weighted  space $   L^2(|\lambda_I| d\overline\lambda \wedge d\lambda)\cap L^\infty $. Consequently, in the proof  
\begin{enumerate}
\item Estimates    near $\kappa_j$, a refined estimate   of the classical Cauchy integral formula,  depends crucially on boundedness of $\widetilde m(x,\lambda)$  near $\kappa_j$ (see \eqref{E:discon}).  These estimates are not valid if $\widetilde m(x,\lambda)$ blows up at $\lambda=\kappa_j$  which is the case for \cite[Eq.(4.11)]{BP302} or \cite[Eq.(41)]{VA04} due to a different boundary condition chosen in \eqref{E:Lax-bdry}. 
\item Introducing by the Sato normalized eigenfunction $\chi(x,\lambda)$ in \eqref{E:sato} and the renormalization Definition \ref{D:reg-m}, $\widetilde m\in W$ has a simple pole at $z_n$, $1\le n\le N$ (see \eqref{E:define-remainder-m-sf-pm2}). Since $T$ is regular at $z_n$ (see \eqref{E:tilde-dbar-m-sf-new-0}), similar analysis as above holds here.
\item  Thanks to \eqref{E:tilde-dbar-m-sf-new-0},  there is a {missing direction in the $\lambda$-plane}, i.e. the real axis,   for the continuous scattering operator $  T$ to decay no matter how smooth the initial data $v_0(x)$ is. Therefore, boundedness of   {$m(x,\lambda)  $} near $\infty$ is vital to derive uniform estimates near $\infty$. In particular, estimates     near $\infty$ can not hold if $m(x,\lambda)$ blows up there, which is the case for   \cite[Eq. (4.9), (4.10), (4.11)]{BP214} due to a different boundary condition chosen in \eqref{E:Lax-bdry}.  
\end{enumerate}
 
\end{proof}

\begin{definition}\label{D:quadrature-hat}
  The  eigenfunction space ${ W }\equiv { W}_{x}$ is the set of functions 
 \begin{align*}
(i)\ & \phi (x, \lambda)=\overline{ \phi (x, \overline\lambda)};\\
(ii)\ & (1-  E_{  \cup_{1\le n\le N}  \widetilde D_{z_n}} )\phi(x, \lambda)\in L^\infty;\\
(iii)\ & \phi(x, \lambda)=\frac{\phi_{z_n,\res}(x )}{\lambda -z_n}  +\phi_{z_n,r}(x, \lambda),\ \lambda \in \widetilde D_{z_n}^\times,
\\
&\phi_{z_n,\res}(x ), \  (\lambda-z_n)\phi_{z_n,r}(x ,\lambda),\ 
   \frac{\phi_{z_n,r}(x, \lambda)}{1+|x| } \in L^\infty( \widetilde D_{z_n}) ;
\\
(iv)\ &     (e^{\kappa_1x_1+\kappa_1^2x_2}\phi (x,\kappa_1^+), \cdots, e^{\kappa_Mx_1+\kappa_M^2x_2}\phi(x,\kappa_M^+)) \widetilde{\mathcal D}=0,\\
  &   \textit{$ \widetilde {\mathcal D}$ is defined by \eqref{E:reg-m-sym}},\ \kappa_j^+= \kappa_j+0^+,\  \frac{\frac\partial{\partial s} \phi (x,\lambda) }{1+|x| } \in L^\infty(\widetilde D_{\kappa_j}).
\end{align*}
\end{definition}

\begin{theorem}\label{T:cauchy-integral-eq}
If $u(x )=u_0(x)+v_0(x )$, $ u_0$ is a totally positive $\textrm{Gr}(N, M)_{> 0}$ KP soliton, $v_0(x)$ is a real valued function, and
\[
 \begin{gathered}
  (1+|x|)^2 \partial_x^kv _0\in  {L^1\cap L^\infty},   | k|\le 4,\  |v_0 |_{{L^1\cap L^\infty}}\ll 1,
 \end{gathered}
 \] then the eigenfunction $\widetilde  m $  derived from Theorem \ref{T:normal-eigen} satisfies 
\beq\label{E:eigen-charac}
\widetilde m (x, \lambda)\in W
 \eeq and the Cauchy integral equation
\beq\label{E:cauchy-integral-equation-sf}
\begin{gathered}
  {\widetilde  m}(x, \lambda) =1+\sum_{n=1}^N\frac{\widetilde  m_{z_n, \res }(x  )}{\lambda -z_n } +\mathcal CT
  \widetilde m ,\ \forall\lambda\ne  0.
\end{gathered}
\eeq  

\end{theorem}

\begin{proof}   
One can adapt the argument as that in the proof of \cite[Theorem 4]{Wu19} to prove the theorem.
 
\end{proof}

In the above theorem, the characterization $ \widetilde m \in W$ is an important companion condition to the Cauchy equation \eqref{E:cauchy-integral-equation-sf}. Observe that, evaluating \eqref{E:cauchy-integral-equation-sf} at $\kappa_j^+$,  we obtain ${M}$ equations 
\begin{eqnarray*}
\widetilde  m (x, \kappa_1^+) = &1+\sum_{n=1}^N \frac{  \widetilde  m_{z_n, \res }(x   ) }  {\kappa_1-z_n}  +\mathcal C_{\kappa_1^+}T
  \widetilde m   ,\\
 \vdots \hskip.5in &\hskip1.2in\vdots \\
\widetilde  m(x, \kappa_M^+) = &1+\sum_{n=1}^N \frac{  \widetilde  m_{z_n, \res }(x   ) }  {\kappa_N-z_n}  +\mathcal C_{\kappa_M^+}T
  \widetilde m     ,  
\end{eqnarray*} 
 for $M+N$ variables $\{ \widetilde  m(x, \kappa_j^+) ,\widetilde  m_{z_n, res }(x   )\} $, $1\le j\le M$, $1\le n\le N$. The $\widetilde {\mathcal D}$ -symmetry in $W$, namely,
\[
(\widetilde \Phi (x,\kappa_1^+), \cdots, \widetilde \Phi(x,\kappa_M^+))  \times \widetilde{\mathcal D}=0
    \]where $\widetilde {\mathcal D}$ is an $M\times N$ matrix, provides us extra $N$ constraints. So $ \widetilde m_{z_n,\res}(x)$ are determined by  discrete scattering data $\kappa_j$, $z_n$, and $\mathcal C_{\kappa_j^+}T\widetilde m$. Along with Theorem \ref{T:linearity}, this closeness property implies that \eqref{E:eigen-charac} and \eqref{E:cauchy-integral-equation-sf}   can serve as a good starting point for the inverse problem.

\end{document}